\documentclass{article}
\usepackage{amsmath,amsthm,amsfonts,amssymb,graphicx,hyperref}
\usepackage[english]{babel}
\usepackage{verbatim}
\usepackage{enumitem}
\usepackage{framed}
\usepackage{stmaryrd,pifont}
\usepackage[numbers]{natbib}
\usepackage{mathtools}
\usepackage{ulem}
\usepackage{cancel} %
\usepackage[utf8]{inputenc} %
\usepackage{xcolor}
\usepackage[all]{xy} \SilentMatrices
\usepackage{tikz}
\usetikzlibrary{matrix,arrows}

\usepackage{graphicx} 
\newcommand{\minus}{\smallsetminus}

\newcommand{\defi}{\mathrel{\mathop:}=} %
\newcommand{\A}{\mathcal{A}}

\newcommand{\Q}{\mathbb{Q}}

\newcommand{\Borel}{\mathcal{B}}
\newcommand{\Power}{\mathcal{P}}
\renewcommand{\emptyset}{\varnothing}

\newcommand{\Rel}{\mathcal{R}}

\newcommand{\Logic}{\mathcal{L}}
\newcommand{\sem}[1]{{\llbracket #1 \rrbracket}} %
\newcommand{\pos}[1]{{\langle #1 \rangle}} %

\newcommand{\lmp}[1]{\mathbb{#1}}

\newcommand{\sm}{\minus}
\newcommand{\sbq}{\subseteq}

\newcommand{\leb}{\mathfrak{m}}

\newcommand{\Inter}{I}

\DeclareMathSymbol{\rest}{\mathord}{AMSa}{"16}

\newcommand{\Tfont}[1]{\mathsf{#1}}

\newcommand{\posbow}[1]{{\langle #1 \rangle_{\bowtie \, q}}}
\newcommand{\s}{\mathsf{s}}
\newcommand{\sfs}{\mathsf{s}}
\newcommand{\e}{\mathsf{e}}
\newcommand{\h}{\mathsf{h}}
\newcommand{\Fam}{\mathcal{F}}
\newcommand{\D}{\mathcal{D}}
\newcommand{\Cat}{\mathbf{C}}
\newcommand{\catsp}{\wedge}
\newcommand{\coalg}{\Delta}
\newcommand{\liftrel}[1]{\widehat{#1}}
\newcommand{\Set}{\mathbf{Set}}
\newcommand{\dom}{\mathrm{dom}}
\newcommand{\id}{\mathrm{id}}
\newcommand{\img}{\mathrm{img}}
\newcommand{\LMP}{\mathbf{LMP}}
\newcommand{\Mes}{\mathbf{Meas}}
\newcommand{\catco}{\vee}
\newcommand{\cmark}{\ding{51}}
\newcommand{\xmark}{\ding{55}}
\DeclareMathOperator{\inl}{inl}
\DeclareMathOperator{\inr}{inr}
\newcommand{\functor}[1]{\mathcal{#1}}

\makeatletter
\let\lemma\@undefined
\let\endlemma\@undefined 
\let\definition\@undefined
\let\enddefinition\@undefined 
\let\example\@undefined
\let\endexample\@undefined 
\let\remark\@undefined
\let\endremark\@undefined
\let\corollary\@undefined
\let\endcorollary\@undefined
\makeatother

\newtheorem{theorem}{Theorem}[section]
\newtheorem{lemma}[theorem]{Lemma}
\newtheorem{prop}[theorem]{Proposition}
\newtheorem{corollary}[theorem]{Corollary}

\newtheorem*{claim*}{Claim}
\theoremstyle{definition}
\newtheorem{definition}[theorem]{Definition}
\newtheorem{remark}[theorem]{Remark}
\newtheorem{question}[theorem]{Question}
\newtheorem{example}[theorem]{Example}
\theoremstyle{remark}

\newtheorem*{example*}{Example}

\newcommand{\xindex}[2][{}]{\relax}

\hypersetup{
  colorlinks,
  urlcolor={blue},
  linkcolor={blue!50!black},
  citecolor={blue!50!black},
}

\begin{document}

\title{A classification of bisimilarities for general Markov decision
  processes%
  \thanks{Universidad Nacional de Córdoba.  Facultad de Matemática, Astronomía,  Física y
    Computación.
    \\
    Centro de Investigación y Estudios de Matemática (CIEM-FaMAF),
    Conicet. Córdoba. Argentina.\\
    Supported by Secyt-UNC project 33620230100751CB and Conicet PIP project 11220210100508CO}
}
\author{Martín Santiago Moroni \and
  Pedro Sánchez Terraf%
}
\maketitle 
\begin{abstract}
  We provide a fine classification of
  bisimilarities between states of possibly different labelled Markov
  processes (LMP). We show
  that a bisimilarity relation proposed by Panangaden that uses direct sums coincides with “event
  bisimilarity” from his joint work with Danos, Desharnais, and
  Laviolette. We also extend Giorgio Bacci's notions of
  bisimilarity  between two different processes to the case of
  nondeterministic LMP and generalize the game characterization of state
  bisimilarity by Clerc et al. for the latter.
\end{abstract}
\section{Introduction}
\label{sec:introduction}

Setting the stage for the study of Markov decision processes over
continuous state spaces requires several mathematical
ingredients. In turn, this branch of Computer Science is a rich
source of appealing research questions in Mathematics, with the plus
of having a well-motivated origin.

In our case of interest, that of \textit{labelled Markov processes}
\cite{Desharnais,DEP}, defining appropriate notions of behavioral
equivalences and working with them involves a non-negligible amount of
descriptive set theory \cite{Kechris} and measure theory, and with the
latter come many “pathologies” related to non-measurable sets. This is
then reflected in the fact that many notions of \textit{bisimilarity}
which coincide when the state space is “regular”, turn out to be
different. One early example is provided by \cite{Pedro20111048},
where it is shown that \textit{event} bisimilarity \cite{coco} does not
coincide with the more traditionally-minded notion of \textit{state}
bisimilarity.

It was also shown in \cite{Pedro20111048} that bisimilarities do not
not interact in a predictable manner with direct sums, and thus the
question whether two states are deemed equivalent might depend on the
\textit{context} in which they are inserted. Therefore, one is
naturally led to a subdivision of the existing notions of
bisimulations into two categories: Those that relate two states of a
single process (“internal” bisimulations) and those relating two
states of different processes (“external” ones).

G.~Bacci's dissertation \cite{bacci} provides an adequate
notion of external (state) bisimilarity, in the sense that it does not
depend on some previous internal definition and thus it does not rely
on sums. One of our contributions concerns the generalization of this
external notion to nondeterministic LMP (NLMP)
\cite{DWTC09:qest,D'Argenio:2012:BNL:2139682.2139685}, and
presenting external versions of event and “hit” bisimilarities as well.

Our main results show that several notions of external bisimilarity
appearing in the literature: coalgebraic bisimilarity; the one given by sums; and that
originating in \cite{bacci}, are pairwise different in
general. Moreover, we show that the third one differs from the one defined using spans.
It is to be noted that all LMP counterexamples used in this paper have
separable state spaces, and hence they can be regarded as subsets of
the real line with the trace of the Borel $\sigma$-algebra.
Aside from these negative results, we show
that a bisimilarity relation proposed by Panangaden
\cite{doi:10.1142/p595} by using direct sums coincides with event
bisimilarity.

The next section reviews some prerequisites, and
Section~\ref{sec:lmp-internal} recalls the original definitions of
bisimulations for LMP and how they behave for the example from
\cite{Pedro20111048} and, in particular, the interaction with zigzag morphisms. Section~\ref{sec:lmp-external} introduces
Bacci's external bisimulation and some results concerning its
interaction with internal state bisimilarity on a sum; several
questions on diverse transitivity situations are addressed there.
We treat categorical bisimulations next, considering spans and
coalgebraic bisimilarity in
Section~\ref{sec:catsp-bisimulations}, and cospans in
Section~\ref{sec:catco-bisimulations}, where a brief study of LMP
quotients and finality takes place.
Other notions based on sums are studied in
Section~\ref{sec:oplus-bisimilarity}, Panangaden's one in particular.
Section~\ref{sec:comparison-notions-lmp} summarizes our current
knowledge in Table~\ref{tab:bisim-comparison}, and
Section~\ref{sec:NLMP} is devoted to NLMP, where we provide a notion of external event
bisimilarity for them and generalize 
the game-theoretic characterization of state bisimilarity obtained in 
\cite{ic19} to those processes. We offer some concluding
remarks in Section~\ref{sec:conclusion}.

\section{Preliminaries}
\label{sec:preliminaries}

Let $R$ be a binary relation on a set $S$. We say that $X \subseteq S$ is
\emph{$R$-closed}
if
\[
  \{s\in S \mid \exists x\in X \; (x\mathrel{R}s \vee 
    s\mathrel{R}x)\} \subseteq X.
\]

If $\Gamma \subseteq \Power(S)$ is a set of subsets of $S$, we denote by $\Gamma(R)$ 
the family of all $R$-closed sets that are elements of $\Gamma$. We can think of $\Gamma(\cdot)$ as an operator that maps relations to subfamilies of $\Power(S)$. We highlight the particular case when \(\Gamma\) is the $\sigma$-algebra of a measurable space $(S,\Sigma)$, and 
note that in such a situation $\Sigma(R)$ is always a sub-$\sigma$-algebra of $\Sigma$. 

Another operator that will be important in this work (which goes in the opposite direction) is the one which associates to each family $\Lambda \subseteq \Power(S)$ a relation $\mathcal{R}(\Lambda)$ on $S$, defined by 
\[
(s,t) \in \mathcal{R}(\Lambda) \text{ if and only if } 
\forall A \in \Lambda \ (s \in A \Leftrightarrow t \in A).
\]
This relation associates points that belong to exactly the same sets in $\Lambda$. We note that it is always an equivalence relation, regardless of the conditions on $\Lambda$.

We construe the “disjoint union” $S\oplus S'$ of two sets as the
tagged union $(S\times \{0\})\cup (S'\times\{1\})$,
and we have the obvious inclusion functions $\inl : S\to S\oplus S'$ and
$\inr : S' \to S\oplus S'$. If $R$ is now a binary relation on $S\oplus S'$,
its restriction to $S$ is defined as
\[
  R\rest S \defi \{ (s_1,s_2) \in S \mid \inl(s_1) \mathrel{R}
    \inl(s_2) \}.
\]
If $\Gamma$ is a family of subsets of $T$ where $S\sbq T$,  its
\emph{restriction} is defined as
$\Gamma\rest S \defi \{Q\cap S \mid Q\in \Gamma \}$. With a slight
abuse of notation, if  $\Gamma \sbq \Power(S\oplus S')$, we set $\Gamma\rest S
\defi \{\inl^{-1}[Q] \mid Q\in \Gamma \}$.

The \emph{sum} of two measurable spaces $(S,\Sigma)$ and
$(S',\Sigma')$ is $(S\oplus S',\Sigma\oplus \Sigma')$, with the
following abuse of notation: $\Sigma\oplus\Sigma'\defi\{ Q\oplus
Q'\mid Q \in \Sigma,\  Q'\in \Sigma'\}$.  We have the following result
on restriction of relations on sums that do not correlate elements in
different components.
\begin{lemma}\label{lem:sigma'_restriction}
  Let $(S,\Sigma)$ and $(S',\Sigma')$ be measurable spaces. If $R$ is
  a binary relation such that $R\subseteq \inl(S)^2 \cup \inr(S')^2$,
  then $(\Sigma\oplus \Sigma')(R)\rest S=\Sigma(R\rest S)$. \qed
\end{lemma}

The set of subprobability measures over a measurable space
$(S,\Sigma)$ will be denoted by $\Delta(S)$.
If $R$ is a relation on $S$, its \emph{lifting} $\bar{R}$ to $\Delta(S)$ is the equivalence relation given by:
\begin{equation}\label{eq:int-lift}
\mu \mathrel{\bar{R}} \mu' \iff \forall Q\in \Sigma(R) \;
\mu(Q)=\mu'(Q).
\end{equation}

For any topological space, $\Borel(X)$ denotes the Borel
$\sigma$-algebra of $X$; for any $V\sbq X$, $\Borel(X)_V$ will stand for
$\sigma(\Borel(X) \cup \{V \})$, the smallest $\sigma$-algebra
containing $V$ and the open sets of $X$. 
We reserve $\Inter$ for the open interval $(0,1)$.
Also, we will use the following notation for standard measures and functions: $\leb$ will denote the Lebesgue measure on $\Inter$ restricted to the Borel sets; $\delta_s(\cdot)$ will denote the Dirac measure with point mass at $x$, and $\chi_{A}(\cdot)$ will be the characteristic function of the set $A$.

The set of nonnegative integers will be denoted by $\omega$.

\section{Labelled Markov processes}
\label{sec:LMP}

\begin{definition}\label{def:Markov-kernel} 
  A \emph{Markov kernel}\xindex{Markov!kernel} on a measurable space
  $(S,\Sigma)$ is a function $\tau:S\times \Sigma \rightarrow [0,1]$
  such that, for each fixed $s \in S$, $\tau(s,\cdot):\Sigma
  \rightarrow [0,1]$ is a subprobability measure, and for each fixed
  set $X \in \Sigma$, $\tau(\cdot,X):S \rightarrow [0,1]$ is a
  $(\Sigma,\Borel([0,1]))$-measurable function.
\end{definition}

These kernels will be interpreted as transition functions in the
processes defined below. Let $L$ be any countable set.

\begin{definition}\label{def:LMP}
  A \emph{labelled Markov process}, or LMP, with
  \emph{label set} $L$, is a triple $\lmp{S}=(S,\Sigma,\{\tau_a \mid a\in L
  \})$, where $S$ is a set, whose elements are called \emph{states},
  $\Sigma$ is a $\sigma$-algebra on $S$, and for all $a\in L$,
  $\tau_a:S \times \Sigma \rightarrow [0,1]$ is a Markov kernel.

  A \emph{pointed} LMP is a tuple $(S,s,\Sigma,\{\tau_a \mid a\in L
  \})$ where $s\in S$ is called the \emph{initial state}.
\end{definition}

The elements of $L$ are also called \emph{actions}. Since $L$ is
mostly fixed throughout each argument, we also write $\{\tau_a \}_a$
instead of $\{\tau_a \mid a\in L \}$.

These processes
were introduced by Desharnais et
al. \cite{Des98LogicalChar,Desharnais,DEP} as a generalization to
continuous spaces of (discrete) probabilistic transition systems. The
operation of an LMP is summarized as follows: At any given moment, the
system is in a state $s\in S$ and, depending on its interaction with
the environment through the labels in $L$, it makes a transition to a
new state according to a probability law (given by the kernels). The term
“Markov” in the name of these processes refers to the fact that
transitions are governed by the current state of the system, not its
past history.

\begin{example}\label{exm:lmp-U}
  The following LMP $\lmp{U}$, (called $\mathbf{S_3}$ in
  \cite{Pedro20111048}) motivated our finer classification of bisimilarities.

  Let $V$ be a Lebesgue non-measurable subset of $\Inter$. By 
  \cite[Thms.~1--4]{Los-extension} (or alternatively,
  \cite[Cor.~11]{Pedro20111048})
  there exist two extensions $\leb_0$ 
  and $\leb_1$ of $\leb$ to $\Borel(\Inter)_V$ such that $\leb_0(V)\neq\leb_1(V)$. 
  Also, for $q\in \Inter\cap \Q$ define $B_q \defi (0,q)$; hence   $\{B_q\mid q\in 
  \Inter \cap\Q\}$ is a countable generating family of $\Borel(\Inter)$.

  Let $s,t,x \notin \Inter $ be mutually distinct; we may view 
  $\leb_0$ and $\leb_1$ as measures defined on the sum $\Inter \oplus 
  \{s,t,x \}$, supported on $\Inter $. The label set will be $L\defi 
  (\Inter \cap \Q)\cup \{\infty \}$. Now define $\lmp{U}=(U,\Upsilon,\{\tau_a 
    \mid a\in L \})$ such that
  \begin{align*}
    (U,\Upsilon)& \defi (\Inter \oplus\{s,t,x \},\Borel(I)_V\oplus 
    \Power(\{s,t,x \})), \\
    \tau_q(r,A) &\defi \chi_{B_q}(r)\cdot \delta_x(A),\\  
    \tau_\infty(r,A)&\defi \chi_{\{s \}}(r)\cdot   
    \leb_0(A)+\chi_{\{t \}}(r)\cdot \leb_1(A)
  \end{align*}
  when $q\in \Inter\cap \Q$ and $A\in \Upsilon$; we have omitted the
  injection maps into the sum to make notation lighter. This defines
  an LMP (see \cite{moroni2020zhou} for further discussion).

  The dynamics of this process goes intuitively as follows: The 
  states $s$ and $t$  can only make an $\infty$-labelled transition  
  to a ``uniformly distributed'' state in $\Inter $, but they 
  disagree on the probability of reaching $V\sbq \Inter$. Then, each 
  point of $B_q\subseteq \Inter$ can  make a $q$-transition to $x$. 
  Finally, $x$ can make no transition at all (see 
  Figure~\ref{fig:lmp-U}).
  \begin{figure}[h]
    \begin{center}
      \includegraphics[scale=0.05]{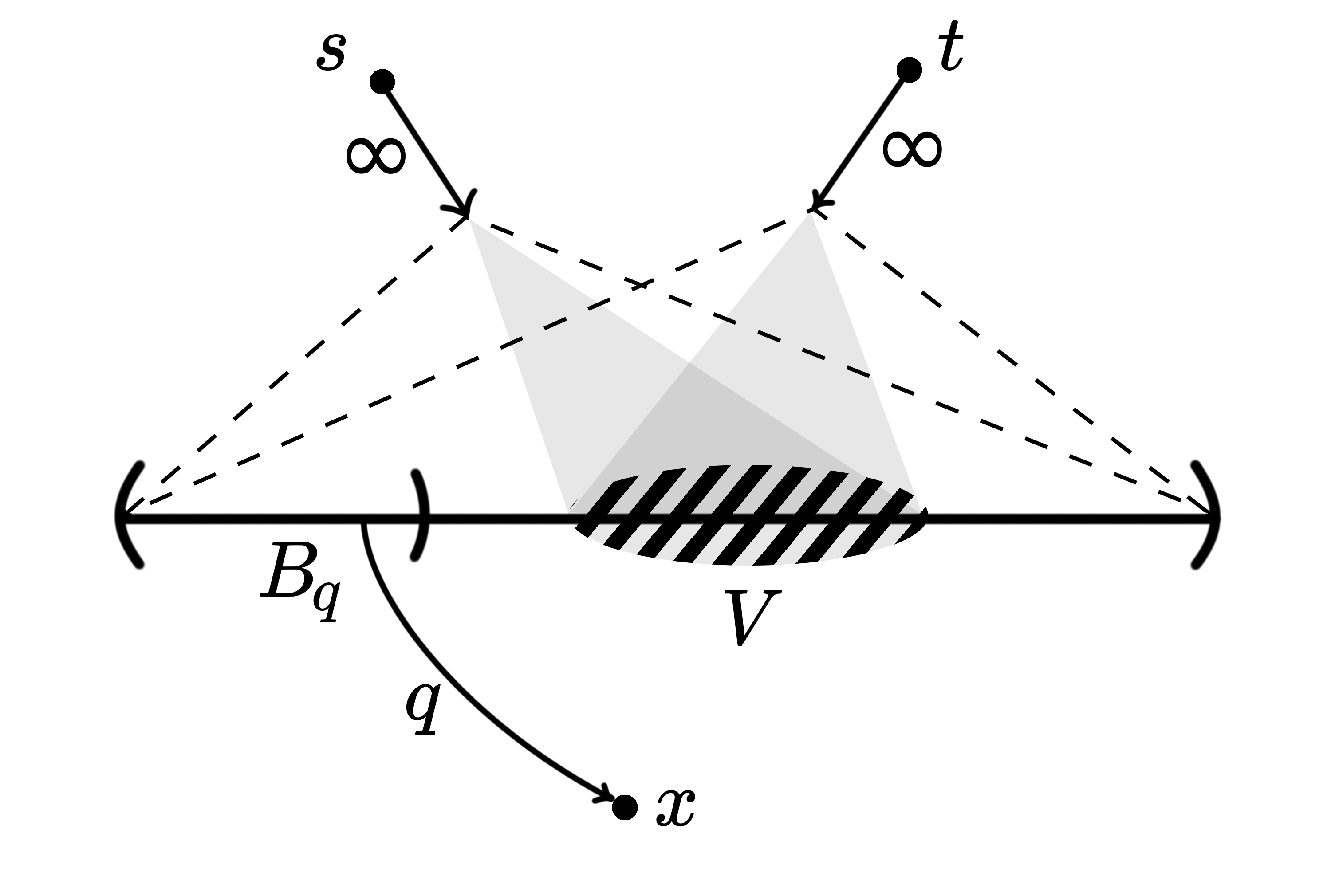}
    \end{center}
    \caption{The LMP $\lmp{U}$.}\label{fig:lmp-U}
  \end{figure}
\end{example}

LMPs form a category  $\mathbf{LMP}$ whose morphisms are given by the
following definition.

\begin{definition}
  A \emph{zigzag} morphism $f$ between two 
  LMPs $\lmp{S}$ and $\lmp{S}'$ is a measurable map $f:S\to{S}'$
  satisfying
  \[
    \forall a \in L, s \in S, B \in \Sigma' \ \ 
    \tau_a(s,f^{-1}[B])=\tau'_a(f(s),B).
  \]
  For pointed LMP, we also require that $f$ preserves the initial
  state.
\end{definition}

\subsection{Internal bisimulations}
\label{sec:lmp-internal}

\subsubsection*{State bisimulations and zigzags}

\begin{definition}\label{def:state-bisimulation}
  Given an LMP $\lmp{S}=(S,\Sigma,\{\tau_a \mid a\in L \})$, an
  \emph{(internal) state bisimulation} $R$ in $\lmp{S}$ is a binary
  relation on $S$ such that if $s\mathrel{R}t$, then for all $C \in
  \Sigma(R)$, it holds that
    \[\forall a\in L \;	\tau_a(s,C)=\tau_a(t,C).\]
    We say that $s$ and $t$ are \emph{state bisimilar}, denoted as
    $s\sim_{\sfs,\lmp{S}} t$, if there exists a state
    bisimulation $R$ such that $s\mathrel{R}t$. 
\end{definition}

We will often drop the subscript $\lmp{S}$ and simply write $s\sim_\sfs t$ if there is no need for clarification.
We have the following well known result:
\begin{prop}\label{prop:union-bisimulation}
  The union of a family of state bisimulations is also a state
  bisimulation. Hence $\sim_\sfs$ is a state bisimulation.
\end{prop}

\begin{example}[State bisimilarity in Example~\ref{exm:lmp-U}]\label{exm:state-in-U}
  We reproduce part of the proof from \cite[Thm.~15]{Pedro20111048}
  showing that state bisimilarity on $\lmp{U}$ is the identity. It is
  straightforward to verify that $\{x \}$ must be an equivalence class
  of $\sim_\sfs$ (and hence  it belongs to $\Sigma(\sim_\sfs)$): For any
  other state $r\in \lmp{U}$, there exists $l\in L$ such that
  $\tau_l(r,U)= 1$ but $\tau_l(x,U)=0$ (and $U$ is obviously
  $\sim_\sfs$-closed). Now, take $y\neq z\in \Inter$; let us show that $y$ and
  $z$ cannot be related by $\sim_\sfs$. 
  Since $\{B_q\mid q\in \Inter\cap\Q\}$
  generates $\Borel(\Inter)$, there exists $a\in \Inter\cap \Q\sbq L$ such that
  $B_a$ separates $y$ from $z$. Without loss of generality, we can
  assume that $\{y,z \}\cap B_a=\{y \}$. Then $\tau_a(y,\{x \})=1$ but
  $\tau_a(z,\{x \})=0$. We conclude that $\sim_\sfs$ restricted to
  $\Inter\cup \{x \}$ is the identity relation, and (in particular) $V\sbq
  \Inter$ is $\sim_\sfs$-closed. Finally, we observe that
  $\tau_\infty(s,V)\neq\tau_\infty(t,V)$, and therefore $s$ and $t$
  are not state bisimilar.
\end{example}

\begin{lemma}\label{lem:zigzag-state-image}
  Let $f:\lmp{S}\to \lmp{S'}$ be a zigzag morphism. If $s\sim_\sfs t$ in
  $\lmp{S}$, then $f(s)\sim_\sfs f(t)$ in $\lmp{S'}$.
\end{lemma}
\begin{proof}
  Let $R$ be a bisimulation such that $s\mathrel{R}t$; it suffices to
  show that 
  $R'\defi f[R]=\{(f(x),f(y))\mid x\mathrel{R}y\}$ is a bisimulation
  on $\lmp{S'}$.
  Suppose that $f(x)\mathrel{R'}f(y)$ with $x\mathrel{R} y$,
  and let $A\in \Sigma'$ be a $R'$-closed set. We want to verify that
  $\tau_a'(f(x),A)=\tau_a(x,f^{-1}[A])$ is equal to
  $\tau_a'(f(y),A)=\tau_a(y,f^{-1}[A])$. For this, it suffices to see
  that $f^{-1}[A]\in \Sigma(R)$ since $x\sim_\sfs y$. The measurability
  is immediate. We verify that $f^{-1}[A]$ is $R$-closed: Let $w\in
  f^{-1}[A]$ and $z\in S$ such that $w\mathrel{R}z$; since
  $f(w)\mathrel{R'} f(z)$, $f(w)\in A$, and $A$ is $R'$-closed, we
  conclude that $f(z)\in A$, i.e., $z\in f^{-1}[A]$. The case
  $z\mathrel{R}w$ is similar.
\end{proof}

\begin{example}\label{exm:zigzag-not-state-domain}
  The converse of Lemma~\ref{lem:zigzag-state-image} is not true:
  Consider the identity function $\id:\lmp{U}\to \lmp{U}_B$ where
  $\lmp{U}_B\defi (U,\Borel(\Inter)\oplus \Power(\{ s,t,x \}),\{ \tau_a \}_a)$.
  Let $\Sigma$ be the $\sigma$-algebra of $\lmp{U}_B$. We see
  that such a function is measurable since $\Sigma\subseteq \Upsilon$
  and it clearly satisfies the zigzag condition. However, in $\lmp{U}_B$,
  we have $s\sim_\sfs t$ as the non-measurable Lebesgue set $V$ is not
  available to distinguish them, whereas $s \nsim_\sfs t$ in $\lmp{U}$.
\end{example}

\subsubsection*{Direct sums}

\begin{definition}\label{def:sum-LMP}
  Given two LMPs $\lmp{S}$ and $\lmp{S}'$, the \emph{direct sum} LMP
  $\lmp{S}\oplus \lmp{S'}$ has the sum of measurable spaces $(S\oplus
  S',\Sigma \oplus \Sigma')$ as underlying measurable space, and
  Markov kernels
  \begin{align*}
    \tau_a^\oplus(\inl(x),A\oplus A') &\defi \tau_a(x,A)  && x\in S,\\
    \tau_a^\oplus(\inr(x),A\oplus A') &\defi \tau_a(x,A')  && x\in S'.
  \end{align*}
\end{definition}

The above construction corresponds to the categorical coproduct or sum in the category $\mathbf{LMP}$. 
It is straightforward from the definition that the following equivalence holds:
\begin{align}
  \label{eq:equiv-oplus-kernels} \tau_a(s,A)= \tau'_a(s',A') \iff \tau^{\oplus}_a(\inl(s),A\oplus A') = \tau^{\oplus}_a(\inr(s'),A\oplus A').
\end{align}

If $R\sbq S\times S$, we can ``lift'' it to relations $R_l$ and 
$R_r$  on $S\oplus S'$  and $S'\oplus S$, respectively as follows:
\begin{equation}
  \label{eq:lift}
  R_l\defi \{(\inl(s),\inl(t))\mid s\mathrel{R}t \},\qquad
  R_r\defi \{(\inr(s),\inr(t))\mid s\mathrel{R}t \}.
\end{equation}
It is easy to verify that if $R$ is a state bisimulation on $\lmp{S}$, then $R_l$ is a state bisimulation on $\lmp{S}\oplus \lmp{S'}$.
This only depends on the fact that $(\Sigma\oplus\Sigma')(R)\rest S=\Sigma(R)$ (see Lemma~\ref{lem:sigma'_restriction}). 
Additionally, note that this implies that state bisimilarity on the
sum of two LMPs restricted to a summand is not smaller than the
original bisimilarity on the latter.

As a warm up for its future use in defining external bisimulations, we
will consider the case of equal summands and relate state bisimilarity
on $\lmp{S}\oplus \lmp{S}$ to state bisimilarity on $\lmp{S}$.
In this setting, we can lift a relation $R$ on $S$ in various
ways:
To the left summand as $R_l$ above, to the right, in a crossed manner, or completely. 
We will focus on the latter case. 
If $R\sbq S\times S$, we define the following binary relation $R^+$ on
$S\oplus S$ (that is, $R^+ \sbq (S\oplus S)\times (S\oplus S)$): 
\begin{equation}\label{eq:complete-lift}
  u \mathrel{R^+} v \iff \exists s,t\in S,\ s \mathrel{R} t \land
  u\in\{\inl(s),\inr(s) \} \land v\in \{\inl(t),\inr(t) \}
\end{equation}
Note that $R^+$ is an equivalence relation if and only if $R$ is.

\begin{lemma}\label{lem:state-bisimulation-in-sum}
  Let $R \sbq S\times S$.
  \begin{enumerate}
  \item $A\sbq S$ is $R$-closed if and only if $A\oplus A$ is
    $R^+$-closed. Additionally, every $R^+$-closed set is of this
    form.
  \item\label{item:state-bisimulation-in-sum} $R \sbq S\times S$ is a
    state bisimulation on $\lmp{S}$ if and only if $R^+$ is a state
    bisimulation on $\lmp{S}\oplus\lmp{S}$.
  \end{enumerate}
\end{lemma}

We have the following (unsurprising) characterization: State bisimilarity on the sum $\lmp{S}\oplus \lmp{S}$ is the complete lifting of bisimilarity on $\lmp{S}$.
For this, we use the following auxiliary result.
\begin{lemma}\label{lem:state-bisim-in-summand}
  ${\sim}_{\sfs, \lmp{S}\oplus \lmp{S}}\rest S = {\sim}_{\sfs, \lmp{S}}$.
\end{lemma}
\begin{proof}
  We are considering the “leftmost” copy of $S$, and hence the lhs equals $\{(x,y)\mid
    \inl(x)\sim_{\sfs, \lmp{S}\oplus \lmp{S}} \inl(y) \}$. 
  For $(\sbq)$ it suffices to show that ${\sim}_{\sfs, \lmp{S}\oplus \lmp{S}} \rest S$ is a
  state bisimulation on $\lmp{S}$.  
  Let $(x,y)\in
  {\sim}_{\sfs, \lmp{S}\oplus \lmp{S}} \rest S$ and $A\in \Sigma({\sim}_{\sfs, \lmp{S}\oplus \lmp{S}} \rest
  S)$.  
  First, we prove that $A\oplus A$ is
  ${\sim}_{\sfs, \lmp{S}\oplus \lmp{S}}$-closed (we already know it is measurable).
  Assume $u\in A\oplus A$ and
  $u\sim_{\sfs, \lmp{S}\oplus \lmp{S}} v$. 
  Choose $s,t\in S$ such that $u\in\{\inl(s),\inr(s) \}$ and $v\in  \{\inl(t),\inr(t) \}$.  
  Then we have
  \[
    \inl(s)\sim_{\sfs, \lmp{S}\oplus \lmp{S}} u\sim_{\sfs, \lmp{S}\oplus \lmp{S}} v\sim_{\sfs, \lmp{S}\oplus \lmp{S}} \inl(t),
  \]
  where the relations in both extremes hold since $(\mathrm{id}_S)^+$  is a bisimulation.   
  Therefore, by transitivity of state bisimilarity we conclude $(s,t)\in
  {\sim}_{\sfs, \lmp{S}\oplus \lmp{S}} \rest S$. 
  Since $A$ is closed for
  ${\sim}_{\sfs, \lmp{S}\oplus \lmp{S}} \rest S$,
  we obtain $t\in A$; hence  $v\in A\oplus A$ and it is
  ${\sim}_{\sfs, \lmp{S}\oplus \lmp{S}}$-closed.
  We have:
  \[
    \tau_a(x,A)
    =\tau_a^\oplus(\inl(x),A\oplus A) =\tau_a^\oplus(\inl(y),A\oplus A)
    =\tau_a(y,A).
  \]
  
  For $(\supseteq)$, recall that  $(\sim_\sfs)_l$ is a
  state bisimulation on $\lmp{S}\oplus \lmp{S}$.
\end{proof}

\begin{prop}\label{prop:state-bisim-in-sum}
	${\sim}_{\sfs, \lmp{S}\oplus \lmp{S}} = ({\sim}_\sfs)^+$.
\end{prop}
\begin{proof}
  The inclusion $(\supseteq)$ follows immediately from
  Lemma~\ref{lem:state-bisimulation-in-sum}(\ref{item:state-bisimulation-in-sum}).
  The other inclusion is equivalent to Lemma~\ref{lem:state-bisim-in-summand}$(\sbq)$.
\end{proof}

\subsubsection*{A review of logic and event bisimulation}

In \cite{Desharnais}, it was shown that for an LMP over an analytic state space, two states are
bisimilar if and only if they
satisfy exactly the same formulas of the following simple modal logic $\Logic$:
\begin{equation}\label{eq:logic} 
  \top \mid \varphi_1 \wedge \varphi_2 \mid \langle a
  \rangle_{>q}\varphi.
\end{equation}
Given an LMP $\lmp{S}$, the semantics $\sem{\varphi}_{\lmp{S}}$ is the set of states of 
$\lmp{S}$ for which the formula holds. For instance, $\sem{\langle a \rangle_{>q}\varphi}$
can be recursively defined to be the set $\{s\in
  S\mid\tau_a(s,\sem{\varphi}) > q \}$. These validity sets are
measurable and 
preserved by zigzag morphisms, as shown in the following lemma (see
\cite[Prop.~9.2]{DEP} for a slightly more general statement):

\begin{lemma}\label{lem:zigzag-logic}
  If $f:\lmp{S}\to \lmp{S}'$ is a zigzag morphism, then for any
  formula $\varphi\in \Logic$, we have
  $f^{-1}[\sem{\varphi}_{\lmp{S}'}] = \sem{\varphi}_{\lmp{S}}$.
  Consequently, $s$ satisfies $\varphi$ in $\lmp{S}$ if and only if $f(s)$
  satisfies $\varphi$ in $\lmp{S}'$.
\end{lemma}
In other words, the truth value of these formulas is preserved
under zigzag morphisms.

In \cite{coco}, a new notion of “event” 
bisimulation is introduced, which shifts
the focus from equivalence relations to the measurable structure of
the LMP. In the same work, it is proved that the bisimilarity associated
with this new concept is characterized by the logic $\Logic$, and this
holds for Markov processes defined on general measurable spaces. The
following is the central concept in this type of bisimulations.

\begin{definition}\label{def:stable}
  Let $\lmp{S}=(S,\Sigma,\{\tau_a \mid a\in L\})$ be an LMP, and let
    $\Lambda\subseteq\Sigma$. We say that $\Lambda$ is
    \emph{stable}\xindex{stable} with respect to $\lmp{S}$ if for all
    $A \in \Lambda$, $r \in [0,1]\cap \Q$, $a \in L$, it holds that
    $\{s\in S \mid \tau_a(s,A)>r\} \in \Lambda$.
\end{definition}

\begin{definition}[{\cite[Def.~4.3]{coco}}]\label{def:event-bisimulation}
  Let $\lmp{S}=(S,\Sigma,\{\tau_a \mid a\in L \})$ be an LMP.  A
  relation $R$ on $S$ will be an \emph{(internal) event bisimulation} if there
  exists a stable sub-$\sigma$-algebra $\Lambda \subseteq \Sigma$ such that $R=\mathcal{R}(\Lambda)$.
  
  Two states $s$ and $t$ of an LMP are \emph{event bisimilar}, denoted
  as $s\sim_{\e,\lmp{S}} t$ (or $s\sim_\e t$), if there exists an
  event bisimulation $R$ such that $s\mathrel{R}t$.
\end{definition}

From Definition~\ref{def:stable}, we can quickly deduce that the union
and intersection of stable families are stable. The next theorem
identifies the smallest stable $\sigma$-algebra. Let $\sem{\Logic}$
denote $\{\sem{\varphi} \mid \varphi \in \Logic \}$.
\begin{theorem}[{\cite[Prop.~5.5]{coco}}] 
\label{thm:logical-characterization}
  Given an LMP $(S,\Sigma,\{\tau_a \mid a\in L \})$,
  $\sigma(\sem{\Logic})$ is the smallest stable $\sigma$-algebra
  contained in $\Sigma$.
\end{theorem}

\begin{corollary}\label{cor:logic-charac-event}
  $\Logic$ characterizes event bisimilarity, i.e.,
  ${\sim_\e}=\Rel(\sem{\Logic})$.
\end{corollary}

\subsection{External bisimulation}
\label{sec:lmp-external}

Up to this point, we have worked with bisimulations defined as relations within a single LMP or underlying space. 
Now we are interested in the case where we have two LMPs, $\lmp{S}$ and $\lmp{S}'$, and we want bisimulations between states of the respective processes.

G.~Bacci's PhD thesis \cite[Sect.~5.1]{bacci} introduces a definition of bisimulation between states of
generalized Markov processes, that is, processes where $\tau_a(s,\cdot)$ is a general measure on $(S,\Sigma)$ rather than a subprobability measure. 
This definition of bisimulation is similar to the internal one, but it requires the analogous concept of an “$R$-closed pair” when $R\subseteq S\times S'$ is a relation between two possibly different sets.

\begin{definition}\label{def:closed-pair}
	Let $R\subseteq S\times S'$ be a relation, and let $A\subseteq S$ and $A'\subseteq S'$. 
	The pair $(A,A')$ is called \emph{$R$-closed pair} if $R\cap(A\times S') = R\cap(S\times A')$.
\end{definition}

We denote with $R[A]$ the set $\{s'\in S' \mid \exists x\in A: x\mathrel{R}s'\}$ (the states in $S'$ related to some state in $A$).
Similarly, $R^{-1}[A']$ is defined as $\{s\in S \mid \exists x'\in A': s\mathrel{R}x'\}$. 
The following lemma gives us a reformulation of the definition of $R$-closed pairs, as well as a basic closure property it satisfies.

\begin{lemma}\label{lem:pair-equiv}
  \begin{enumerate}
  \item\label{item:in-equiv} The following are equivalent:
    \begin{itemize}
    \item
      $(A,A')$ is an $R$-closed pair;
    \item
      $R[A]\subseteq A'$ and $R^{-1}[A']\subseteq A$;
    \item
      $\forall (s,s')\in R: s\in A \Leftrightarrow s'\in A'$.
    \end{itemize}
  \item The family of $R$-closed pairs is closed under complementation
    and arbitrary unions and intersections in coordinates.
  \item\label{item:closed-pairs-antimonotonicity} If $R_0\subseteq
    R_1$, then every $R_1$-closed pair is also $R_0$-closed.
  \end{enumerate}
\end{lemma}

When the two sets are equal, the notions of $R$-closed pair and $R$-closed set coincide in the following sense.

\begin{lemma}\label{lem:closed-vs-closed-pair}
	If $R\subseteq S\times S$, the following statements hold:
	\begin{enumerate}
		\item $A$ is $R$-closed if and only if $(A,A)$ is an $R$-closed pair.
		\item\label{item:closed-pair-reflexive-case} If $R$ is reflexive and $(A,A')$ is an $R$-closed pair, then $A=A'$.
	\end{enumerate}
\end{lemma}

If we have measurable spaces $(S,\Sigma)$ and $(S',\Sigma')$, we say that $(A,A')$ is a \emph{measurable pair} if $A\in \Sigma$ and $A'\in \Sigma'$.

\begin{definition}[{\cite[Def.~5.1.5]{bacci}}]\label{def:ext-bisimulation}
	An \emph{external bisimulation} (or $\times$-bisimulation) between two LMPs $\lmp{S}$ and $\lmp{S}'$ is a relation $R\subseteq S\times S'$ such that if $s\mathrel{R}s'$ and $(A,A')$ is a $R$-closed measurable pair, then 
	\[ \forall a\in L \; \tau_a(s,A)=\tau_a'(s',A').\]
	Two states $s\in S$ and $s'\in S'$ are externally bisimilar, denoted by $s \sim^{\times}_{\lmp{S},\lmp{S'}} s'$, if there exists an external bisimulation relating them. Two pointed LMPs are externally bisimilar if their initial states are externally bisimilar.
\end{definition}
Whenever possible, we simplify the notation and write $\sim^{\times}$
when the processes are understood from the context.

To take advantage of the information we have collected about internal bisimulations, we can consider relations between $\lmp{S}$ and $\lmp{S}'$ as relations on $\lmp{S}\oplus\lmp{S}'$. 
The following proposition proves that there is a correspondence
between $\times$-bisimulations and  internal bisimulations in the sum,
as long as the latter are reduced to subsets of the product $S\times S'$. 
We can specify this with the following notation: 
Given a relation $R$ on the sum $S\oplus S'$, we define its \emph{descent}:
\begin{equation}\label{eq:R_times}
	R_\times\defi \{(s,s')\mid \inl(s) \mathrel{R} \inr(s')\}\sbq S\times S'.  
\end{equation}

If now $R\sbq S\times S'$, we can lift it to a relation on $S\oplus S'$ as follows:
\begin{equation*}
	\liftrel{R}\defi \{(\inl(s),\inr(s'))\mid s\mathrel{R} s'\} 
	\sbq (S\oplus S')\times (S\oplus S').
\end{equation*}

\begin{prop}\label{prop:ext-equiv-oplus}
	\begin{enumerate}
		\item\label{item:ext-equiv-oplus-closed} Let $R$ be a relation on $S\oplus S'$. 
		If $A\oplus A'$ is $R$-closed, then the pair $(A,A')$ is $R_\times$-closed.
		Conversely, if $R\sbq S\times S'$ and the pair $(A,A')$ is $R$-closed, then $A\oplus A'$ is $\liftrel{R}$-closed.
		\item\label{item:ext-equiv-oplus-bisim} If $R\sbq S\times S'$ is a $\times$-bisimulation, then $\liftrel{R}$ is an internal bisimulation on $\lmp{S}\oplus \lmp{S'}$. 
		Conversely, if $R$ is an internal bisimulation on $\lmp{S}\oplus \lmp{S'}$ such that $R\sbq \liftrel{R_\times}$, then $R_\times$ is a $\times$-bisimulation.    
	\end{enumerate}
\end{prop}
The hypothesis $R\sbq \liftrel{R_\times}$ in the second item is another way
of stating that $R$ does not relate points within the same summand.

As we intended, we can now leverage this correspondence (modulo types) and deduce properties of external bisimulations using the known information about internal ones. 
For example, we have an alternative proof of \cite[Prop.~5.1.7]{bacci}:

\begin{corollary}
  The union of a family of $\times$-bisimulations is a $\times$-bisimulation. 
  In particular, $\sim^\times$ is.
\end{corollary}
\begin{proof}
	Given a family $\Fam$ of $\times$-bisimulations between $\lmp{S}$ and $\lmp{S}'$, by Proposition~\ref{prop:ext-equiv-oplus}(\ref{item:ext-equiv-oplus-bisim}), we only need to verify that $\bigcup\{\liftrel{R}\mid R\in\Fam\}$ is an internal bisimulation on $\lmp{S}\oplus\lmp{S}'$ that only relates states in separate processes. 
	This is straightforward from Proposition~\ref{prop:union-bisimulation} and the definition of $\liftrel{R}$.
\end{proof}

We will use $(\sim_\sfs)_\times$ as an abbreviation of $(\sim_{\sfs,\lmp{S}\oplus\lmp{S}'})_\times$.

\begin{corollary}\label{cor:ext-bis-sbq-R-oplus}
	${\sim}^\times\sbq ({\sim}_\sfs)_\times$.
\end{corollary}

This Corollary allows us to obtain an external version of
Lemma~\ref{lem:zigzag-state-image}:

\begin{prop}\label{prop:ext-bisim-any-cospan}
  Assume two zigzag morphisms $\lmp{S} \xrightarrow{f} \lmp{T}
  \xleftarrow{g} \lmp{S}'$, $s\in \lmp{S}$, $s'\in \lmp{S}'$, and
  $s\sim^\times s'$. Then $f(s) \sim_{\sfs,\lmp{T}} g(s')$.
\end{prop}
\begin{proof}
  Consider the diagram $\lmp{S} \xrightarrow{\inl} \lmp{S}\oplus \lmp{S}' \xleftarrow{\inr} \lmp{S}'$ in $\LMP$ and recall that this is the coproduct of $\lmp{S}$ and $\lmp{S}'$ in this category.
  Thus, we know that there exists a morphism $\alpha: \lmp{S}\oplus\lmp{S}' \to \lmp{T}$ such that $\alpha \circ \inl=f$ and $\alpha\circ \inr = g$.
  By Corollary~\ref{cor:ext-bis-sbq-R-oplus} $s\sim^\times s'$ implies $\inl(s)\sim_{\sfs,\lmp{S}\oplus \lmp{S}'} \inr(s')$, and Lemma~\ref{lem:zigzag-state-image} guarantees that $f(s)= \alpha(\inl(s))\sim_{\sfs,\lmp{T}} \alpha(\inr(s'))=g(s')$.
\end{proof}

The question of whether $({\sim}_\sfs)_\times\sbq {\sim}^\times$ will resurface in a different context later on. 

In the case where $\lmp{S}=\lmp{S}'$, we can directly compare external and internal bisimulations because they are of the same type. 
In this situation, Lemma~\ref{lem:closed-vs-closed-pair} guarantees that every $\times$-bisimulation is an internal bisimulation, and both concepts coincide for reflexive relations.

\begin{prop}\label{prop:same-LMP-ext-bisi-equal-internal}
	If $\lmp{S}=\lmp{S}'$, the equalities ${\sim^\times}={\sim_\sfs} = ({\sim}_\sfs)_\times$ hold.
\end{prop}

\begin{proof}
	As $\sim_\sfs$ is an equivalence relation, it is also a $\times$-bisimulation, therefore ${\sim_\sfs}\sbq {\sim^\times}$. 
	For the converse inclusion, the $\times$-bisimulation $\sim^\times$ is also an internal bisimulation, and thus included in $\sim_\sfs$.
	The last equality follows from Proposition~\ref{prop:state-bisim-in-sum} and the identity ${\sim}_\sfs = (({\sim_\sfs})^+)_\times$.
\end{proof}

An immediate consequence is that in this case of a single process, $\sim^\times$ is an equivalence relation \cite[Thm.~8]{Bacci2014BisimulationOM}. 
This information, combined with that provided by Lemma~\ref{lem:state-bisimulation-in-sum} and Proposition~\ref{prop:ext-equiv-oplus}, when specialized to the case $\lmp{S}=\lmp{S}'$, shows that there are no major differences between internal, external, and internal-in-sum bisimulations, as long as they meet some straightforward regularity conditions. 
Remember that every state bisimulation on an LMP $\lmp{S}$ is included in a state bisimulation that is also an equivalence relation \cite[Prop.~4.12]{coco}. %

We comment briefly on zigzag morphisms here. As expected, these functions provide a way to obtain external bisimulations.

\begin{lemma} \label{lem:graph(f)-bisim}
	If $f:\lmp{S}\to \lmp{S'}$ is a zigzag morphism, then $\mathrm{graph}(f)=\{(s,f(s))\mid s \in S\}$ is a $\times$-bisimulation. Conversely, if $f$ is a measurable, surjective function, and $\mathrm{graph}(f)$ is a $\times$-bisimulation, then $f$ is a zigzag morphism.
\end{lemma}

With this lemma, we recover the following result from
\cite[Prop.~3.5.3]{Desharnais}, where it is stated in a slightly
different way involving “$\oplus$-bisimulations”, to be reviewed in
Section~\ref{sec:oplus-bisimilarity}.
\begin{corollary}\label{coro:prop353-Desharnais}
	If $f:\lmp{S}\to \lmp{S'}$ is a zigzag morphism, then $s\sim^{\times}f(s)$.
\end{corollary}

\subsubsection*{Mutual transitivities}

The question about transitivity for different notions of bisimilarity is not generally a straightforward one. 
For example, in the case of internal bisimulation of pointed LMPs as in \cite{Desharnais}, a logical characterization in analytic spaces is provided, which solves this question. 
In \cite[Prop.~7.12]{doi:10.1142/p595}, a proof of transitivity is given for the modified definition of $\oplus$-bisimilarity that we will see in Section~\ref{sec:oplus-bisimilarity}. 
For the categorical definitions that we will see in Section~\ref{sec:catsp-bisimulations}, the same question requires highly technical tools to be answered only in certain cases.

We can ask a more general question about the relationship between one
notion of bisimulation and the composition of potentially different
ones. The first known fact of this sort involving external
bisimilarity concerns the “$z$-closure” of a relation. 
A relation $R\sbq S\times S'$ is \emph{$z$-closed} if, for all $x,y\in S$ and $x',y'\in S'$, $(x,x'),(y,x'),(y,y')\in R \implies (x,y')\in R$. 
The \emph{$z$-closure} of a relation is the smallest $z$-closed relation that contains it. 
In \cite[Lem.~5.3.5]{bacci} it is proven 
that if $R$ is a $\times$-bisimulation between $\lmp{S}$ and $\lmp{S}'$, then its $z$-closure is also a $\times$-bisimulation.
From this, it can be deduced that $\sim^\times$ is $z$-closed: If $x\sim^\times x'$, $y\sim^\times x'$, and $y\sim^\times y'$, then $x\sim^\times y'$.

One interesting question regarding transitivity involves three LMPs,
as in the following statement:
\begin{equation}\label{eq:times-transitivity}
	s \sim^\times_{\lmp{S}_0,\lmp{S}_1} s' \;\wedge \; 
	s'\sim^\times_{\lmp{S}_1,\lmp{S}_2} s'' \implies 
	s\sim^\times_{\lmp{S}_0,\lmp{S}_2} s''.
\end{equation}

The structures defined in \cite[Exm.~1]{Pedro20111048} provide a counterexample to \eqref{eq:times-transitivity}. We reproduce it below with some modifications in the notation, as it will be useful on other occasions.

\begin{example}\label{exm:ext-bisim-not-transitive}
	We consider the LMP $\lmp{U}$ from Example~\ref{exm:lmp-U} and the following modifications:
	\begin{align*}
		\lmp{U}_s&=(U\setminus\{t\},\Upsilon\rest(U\setminus\{t\}),\{\tau_a\mid a\in L\}), \\
		\lmp{U}_t&=(U\setminus\{s\},\Upsilon\rest(U\setminus\{s\}),\{\tau_a\mid a\in L\}),\\
		\lmp{T}'&=((U\setminus \{s,t\})\cup\{t'\},\sigma(\Borel(\Inter)\cup \Power(\{t',x\})),\{\bar{\tau}_a\mid a\in L\}).
	\end{align*}
	Here, we understand that the kernels in $\lmp{U}_s$ and $\lmp{U}_t$ are the appropriate restrictions. The kernels $\bar{\tau}_a$ coincide with $\tau_a$ if $a\in \Inter\cap \Q$, and for $E\in \sigma(\Borel(\Inter)\cup\Power(\{t',x\}))$,
	\[\bar{\tau}_\infty(t',E)= \leb(E), \qquad \bar{\tau}_\infty(r,E)= 0 \text{ for } r\neq t'.\]
	The fact that the non-measurable set $V$ is not available in $\lmp{T}'$ allows for the bisimilarity $s\sim^\times_{\lmp{U}_s,\lmp{T}'} t'$. Additionally, it is clear that $t'\sim^\times_{\lmp{T}',\lmp{U}_t} t$. However, $s \nsim^\times_{\lmp{U}_s,\lmp{U}_t} t$. 
	If they were bisimilar, $\inl(s)$ and $\inr(t)$ would be state bisimilar in $\lmp{U}_s\oplus\lmp{U}_t$, and the proof that they are not closely resembles the one given in Example~\ref{exm:state-in-U}, where the state bisimilarity is explicitly calculated.
\end{example}

If two LMPs coincide in the implication \eqref{eq:times-transitivity}, we know that $\sim_{\sfs,\lmp{S}}$ and $\sim^\times_{\lmp{S},\lmp{S}}$ coincide; therefore, the question can be understood as a connection between the internal bisimilarity on one of them and the external bisimilarity between the two.
If $\lmp{S}_0=\lmp{S}_2$, we obtain the implication $x\sim^\times x' \; \wedge \; y\sim^\times x' \implies x\sim_\sfs y$.
It is not hard to see that this particular case also fails; a counterexample involves the LMP $\lmp{U}_B$ (Example~\ref{exm:zigzag-not-state-domain}) obtained by modifying the $\sigma$-algebra of $\lmp{U}$ and allowing only Borel sets in the interval. 
Between $\lmp{U}$ and $\lmp{U}_B$, it holds that $s\sim^\times s$ and $t\sim^\times s$, but $s \nsim_\sfs t$ in $\lmp{U}$. 
The failure of this property is reasonable because an external bisimulation does not impose conditions on what happens within one of the two LMPs.

The third type of transitivity corresponds to the case $\lmp{S}_0=\lmp{S}_1$, namely,
\begin{question}\label{question:transitivity}
	$y\sim_\sfs x \; \wedge \; x\sim^\times x'\implies 
	y\sim^\times x'$.
\end{question} 

One attempt to prove such a property proceeds by considering the direct sum $\lmp{S}\oplus \lmp{S}'$. 
According to Proposition~\ref{prop:ext-equiv-oplus}(\ref{item:ext-equiv-oplus-bisim}), $\liftrel{{\sim}^\times}$ is an internal state bisimulation on $\lmp{S}\oplus \lmp{S}'$. 
On the other hand, we recall that the left lift %
$(\sim_{\sfs,\lmp{S}})_l$ is also a state bisimulation on $\lmp{S}\oplus \lmp{S}'$. 
Furthermore, it is easy to see that the smallest equivalence relation containing these two internal bisimulations is an internal bisimulation too.
We could conclude the proof if we show that its descent %
is a $\times$-bisimulation. 
However, the following example serves to discard the possibility of proving that if $R$ is an internal bisimulation on $\lmp{S} \oplus\lmp{S}'$ then $R_\times$ is a $\times$-bisimulation.

\begin{example}\label{exm:restr-not-bisim}
	Let $\lmp{S}$ be the discrete process consisting of two states, $x$ and $y$, with a single transition from $x$ to $y$ with probability 1. 
	Let $\lmp{S}'$ be a copy of $\lmp{S}$ with an additional third isolated state.
	These processes can be represented as in the diagrams
        \begin{center}
	  \begin{tabular}{rc@{\hspace{5em}}rr}
	    $\lmp{S}$: & $x \longrightarrow y$  & $\lmp{S}'$: & $x'
            \longrightarrow y'$\\
            &  &  & \rule{0em}{3ex} $z'$ 
          \end{tabular}
        \end{center}
	Consider the following relation $R$ on $S\oplus S'$ and its descent $R_\times$:
	\begin{align*}
		R &= \{(\inl(x),\inr(x')), (\inl(y),\inr(z')), (\inr(y'),\inr(z'))\}, \\
		R_\times &= \{(x,x'), (y,z')\}.
	\end{align*} 
	Then, $R$ is an internal bisimulation on $\lmp{S}\oplus \lmp{S}'$, since the only nontrivial $R$-closed sets are $\{\inl(x), \inr(x')\}$ and $\{ \inl(y), \inr(y'), \inr(z')\}$.
	However, $R_\times$ is not a $\times$-bisimulation between $\lmp{S}$ and $\lmp{S}'$, because the pair $(\emptyset,\{y'\})$ is $R$-closed measurable, but $\tau(x,\emptyset)= 0\neq 1 = \tau'(x',\{y'\})$.
\end{example}

Note that in this example, the relation $R$ is not 
an equivalence relation.
This raises the following question:

\begin{question}\label{question:equiv-int-sum-implies-ext}
  Suppose $R$ is an equivalence and an internal bisimulation on the sum. %
  Does there exist a $\times$-bisimulation that contains $R_\times$?
\end{question}

By Proposition~\ref{prop:ext-equiv-oplus}(\ref{item:ext-equiv-oplus-bisim}) we already know that for bisimulations of the correct type in the sum (which are not equivalence relations), the answer is positive. 
The question now is whether we still have this property for richer relations.
In the particular case of $R={\sim}_{\sfs,\lmp{S}\oplus \lmp{S}'}$, by
Corollary~\ref{cor:ext-bis-sbq-R-oplus}, we know that
$({\sim}_{\sfs,\lmp{S}\oplus \lmp{S}'})_\times = ({\sim}_\sfs)_\times \supseteq
{\sim}^\times$. Therefore, each of the following statements are
equivalent to a positive answer to
Question~\ref{question:equiv-int-sum-implies-ext}:
\begin{enumerate}
\item
  There is a $\times$-bisimulation that contains $({\sim}_\sfs)_\times$;
\item
  $({\sim}_\sfs)_\times = {\sim^\times}$;
\item
  $\liftrel{({\sim}_\sfs)_\times}$ is a state bisimulation on $\lmp{S}\oplus \lmp{S}'$.
\end{enumerate}

\begin{prop}\label{prop:nec-condition-transitivity}
	Let $\lmp{S}$ and $\lmp{S}'$ be two LMPs such that $({\sim}_\sfs)_\times = {\sim^\times}$.
	If $R\sbq S\times S$ is a state equivalence, and $R'\sbq S\times S'$ is a $\times$-bisimulation, then $R\circ R' \sbq {\sim^\times}$, where $s\mathrel{R\circ R'}s' \iff \exists t\in S \ (s\mathrel{R} t \,\wedge\, t\mathrel{R'}s')$. As a consequence, the transitivity of Question~\ref{question:transitivity} holds.
\end{prop}

Corollary~\ref{cor:R-oplus-ext} states a particular case where the equality $({\sim}_\sfs)_\times = {\sim}^\times$ holds.

We conclude this section with a brief comment on the possibility of
defining a processes  “generated from a state.”
Suppose we have a single LMP $\lmp{S}$ and two states $s_0,s_1\in S$ such that $s_0\sim_\sfs s_1$. 
Our motivation lies in obtaining two LMPs $\lmp{S}_0$ and $\lmp{S}_1$ such that $s_0\in S_0$, $s_1\in S_1$, and $s_0\sim^\times s_1$. 
One trivial way to achieve this, as a consequence of Lemma~\ref{lem:closed-vs-closed-pair}, is to duplicate the LMP $\lmp{S}$.
With this approach, the contexts of those states are exactly the same
as before.
A less trivial approach is to try to “minimize” the context of a state
$s$ by defining an LMP \textit{generated} by it.
If $X_s$ were the base space of such an LMP, we would be interested in ensuring that it satisfies the following properties:
\begin{enumerate}
\item $s\in X_s$;
\item the probability of each trace starting in $s$ does not decrease in
  comparison to the original LMP;
\item\label{item:outside-Xs-irrelevant} what remains outside $X_s$ must be irrelevant, i.e., $\tau_a(r,S\sm X_s)=0$ for all $a\in L$ and $r\in X_s$.
\end{enumerate}

We cannot define $X_s$ as a $\sbq$-minimal set for these properties
because oftentimes arbitrary points can be removed without modifying
probabilities (e.g., when all measures are continuous).
Thus we can use a less strict concept of a sub-LMP, defined as
follows: Given $B\sbq S$ which is a \textit{thick} subset of all the
measure spaces $(S,\Sigma,\tau_a(x,\cdot))$
(cf. \cite[Prop.~4.34]{moroniphd}), we define the process $\lmp{S}\rest B = (B,\Sigma\rest B,\{\tau_a\rest_{B\times \Sigma\rest B}\mid a\in L\})$.
Thickness ensures that a sub-LMP containing $s$ will include all the relevant probabilistic information.
This way, Item~\ref{item:outside-Xs-irrelevant} above is guaranteed
whenever $X_s\defi B$ is measurable.
However, even with this idea, it may happen that $s,t \in S$ are state
bisimilar in $\lmp{S}$ but they are not related by any external
bisimulation between the respective generated LMPs. 
A case of this will be shown in Example~\ref{exm:ext-not-cat} in the next section.

\subsection{$\catsp$-bisimulations and coalgebraic bisimilarity}
\label{sec:catsp-bisimulations}

\subsubsection*{Categorical bisimilarity}

A span between two objects $X, Y$ in a category $\Cat$ consists of a third object $Z$ and two morphisms $f:Z\to X$ and $g:Z\to Y$; we will gather all of its components into a tuple as in $(Z,f,g)$
This serves as a generalization of a binary relation $Z$, thought of as a diagram $X \xleftarrow{\pi_X} Z \xrightarrow{\pi_Y} Y$ in the category $\Set$ (a ``jointly monic span"), and it can be used to give a more general categorical definition of bisimulation.
Dually, a cospan between $U$ and $V$ consists of an object $W$ and two morphisms with the same codomain $j:U \to W$ and $k:V\to W$.

\begin{definition}\label{def:catsp-bisimulation}
	A \emph{categorical bisimulation} (or $\catsp$-bisimulation) between two LMPs $\lmp{S}$ and $\lmp{S}'$ is a span $(\lmp{T},f,g)$ in the category $\LMP$.
	We say that $s\in S$ and $s'\in S'$ are \emph{$\catsp$-bisimilar} if there exists a span $(\lmp{T},f,g)$ such that $f(t)=s$ and $g(t)=s'$ for some $t$ in $T$. 
	In such case, we write $s\sim^{\catsp} s'$.
\end{definition}

This definition is presented in \cite{BDEP-IEEE}, where zigzag morphisms include the surjectivity condition, which we omit in this work.

The following result indicates that the existence of a span allows us to generate an external state bisimulation.

\begin{prop}\label{prop:span-implies-ext}
	If there is a span $\lmp{S} \xleftarrow{f} \lmp{W} \xrightarrow{g} \lmp{S'}$, then for every $w\in W$, $f(w)\sim^{\times} g(w)$.
\end{prop}
\begin{proof}
	It is easy to see that the relation $R=\{(s,s')\mid \exists w\in W \; s=f(w) \wedge s'=g(w)\}=\{(f(w),g(w))\mid w\in W\} =(f\times g)[W] \sbq S\times S'$ is a $\times$-bisimulation.
\end{proof}

The proof of this proposition generalizes the first part of Lemma~\ref{lem:graph(f)-bisim}, and if we focus on states, it means that ${\sim}^\catsp\sbq {\sim}^\times$. 
Additionally, it also shows a case where the transitivity~\eqref{eq:times-transitivity} holds: $f(w)\sim^\times w\sim^\times g(w)\implies f(w)\sim^\times g(w)$.

The following example shows that the converse is not true in the following sense: There can be two LMPs in which each state from the first is externally bisimilar to one from the second and vice versa, yet no span exists that preserves that relation.

\begin{example}\label{exm:ext-not-cat}
	We revisit the LMP $\lmp{U}$ and the modifications of Example~\ref{exm:ext-bisim-not-transitive}. 
	In the LMPs $\lmp{U}_s$ and $\lmp{U}_t\oplus\lmp{T'}$, every element of the first is externally bisimilar to one from the second and vice versa. 
	Indeed, the relation $R$ that identifies the states $s$, $t$ and $t'$ and also identifies the three copies of each state in $\lmp{U}_s$ and $\lmp{U}_t\oplus\lmp{T'}$ is a bisimulation relation on the sum, as proven in \cite[Exm.~1]{Pedro20111048}. 
	Moreover, it is the bisimilarity $\sim_{\sfs, \lmp{U}_s\oplus (\lmp{U}_t\oplus \lmp{T}')}$, and by Corollary~\ref{cor:R-oplus-ext}, its descend to $\lmp{U}_s \times (\lmp{U}_t\oplus\lmp{T'})$ is the external bisimilarity.
	However, we will prove that there is no span $\lmp{U}_s \xleftarrow{f} \lmp{W} \xrightarrow{g} \lmp{U}_t\oplus\lmp{T'}$ such that $f(w_0)=s$ and $g(w_0)=t$ for some $w_0$ in $W$.
	
	Assume, in search of a contradiction, that such a span exists and define $\bar{W}\defi W\setminus g^{-1}[T']$ and the sub-LMP $\bar{\lmp{W}}\defi (\bar{W},\Sigma_W\rest{\bar{W}}, \{\tau^W_a\rest_{\bar{W}\times\Sigma_W\rest{\bar{W}}}\})$. 
	We can then restrict the morphisms and consider $f\rest_{\bar{W}}:\bar{W}\to \lmp{U}_s$ and $g\rest_{\bar{W}}:\bar{W}\to \lmp{U}_t$. 
	We will prove that these restrictions are also zigzag morphisms.
	Their measurability follows from the fact that $\bar{W}$ is measurable. 
	Now, let us verify that $f\rest_{\bar{W}}$ is zigzag: If $A\subseteq \lmp{U}_s$ is measurable, and $\bar{w}\in \bar{W}$, we want to show that $\tau^W_a(\bar{w},f^{-1}[A]) = \tau^W_a(\bar{w},(f\rest_{\bar{W}})^{-1}[A])$. 
	We write $f^{-1}[A] = (f^{-1}[A]\cap\bar{W})\cup(f^{-1}[A]\cap g^{-1}[T']) = (f\rest_{\bar{W}})^{-1}[A]\cup(f^{-1}[A]\cap g^{-1}[T'])$. 
	Since $f^{-1}[A]\cap g^{-1} [T']\subseteq g^{-1}[T']$ and $g(\bar{w})\notin T'$, we have $\tau^W_a(\bar{w},f^{-1}[A]\cap g^{-1}[T'])\leq \tau^W_a(\bar{w},g^{-1}[T']) =\tau_a(g(\bar{w}),T')=0$.
	Hence, $\tau_a(f\rest_{\bar{W}}(\bar{w}),A) =\tau_a(f(\bar{w}),A)= \tau^W_a(\bar{w},f^{-1}[A]) = \tau^W_a(\bar{w},(f\rest_{\bar{W}})^{-1}[A])$, and this proves that $f\rest_{\bar{W}}$ is zigzag.
	 
	The case for $g\rest_{\bar{W}}$ is a bit simpler: If $B\subseteq \lmp{U}_t$ is measurable, then $g^{-1}[B] =(g\rest_{\bar{W}})^{-1}[B]$ and thus 
	\[\tau_a(g\rest_{\bar{W}}(\bar{w}),B) = \tau_a(g(\bar{w}),B) = \tau^W_a(\bar{w},g^{-1}[B]) = \tau^W_a(\bar{w},(g\rest_{\bar{W}})^{-1}[B]),\] 
	proving that this restriction is also zigzag.
	
	We have constructed a span $\lmp{U}_s 
	\xleftarrow{f\rest_{\bar{W}}} \lmp{\bar{W}} 
	\xrightarrow{g\rest_{\bar{W}}} \lmp{U}_t$ such that 
	$f\rest_{\bar{W}}(w_0)=s$ and $g\rest_{\bar{W}}(w_0)=t$.
	However, this is not possible because Proposition~\ref{prop:span-implies-ext} would imply that $s\sim^\times t$, which contradicts what we observed in Example~\ref{exm:ext-bisim-not-transitive}. 
\end{example}

\subsubsection*{Coalgebraic bisimilarity}

A stronger version of $\catsp$-bisimilarity can be defined by
appealing to coalgebraic concepts.
LMPs can be presented as coalgebras for the $\Mes$-endofunctor
$\Delta^L$ (a power of the Giry monad \cite{Giry}) defined on objects
as the measurable space
\[
  \Delta^L (S,\Sigma) \defi (\{ \theta \mid
    \theta : L \to \Delta(S) \}, \Delta^L(\Sigma))
\]
where  $\Delta^L(\Sigma)$ is the initial $\sigma$-algebra for the family of
maps
\[
  \{ \theta \mapsto \theta(a)(E) \mid E\in \Sigma, \; a\in L \}.
\]
In this setting, an LMP $(S,\Sigma,\{\tau_a \}_a)$ corresponds to
the coalgebra structure $\tau : (S,\Sigma) \to \Delta^L (S,\Sigma)$
given by $\tau(a)(Q) \defi \tau_a(Q)$ for $a\in L$ and $Q\in
\Sigma$. Hence zigzag  morphisms correspond \textit{mutatis mutandis}
to those of $\Delta^L$-coalgebras.

The definition of bisimulation between coalgebras was introduced by
Aczel and Mendler \cite{am89} and adapted to our context of interest
by de Vink and Rutten \cite{Vink1997BisimulationFP}.

\begin{definition}\label{def:coalgebraic-bisimulation}
	A \emph{coalgebraic bisimulation} or \emph{$\coalg$-bisimulation} between two $\Delta^L$-coalgebras $\lmp{S}$ and $\lmp{S}'$ is a relation $R\sbq S\times S'$ such that there exists a structure of $\Delta^L$-coalgebra $\gamma:R\to \Delta^L(R)$ that makes the projections $\pi:R\to S$ and $\pi':R\to S'$ coalgebra morphisms, i.e., it makes the following diagram commute:
	\[\xymatrix{
		S \ar[d]_\tau & R \ar[l]_{\pi} \ar[r]^{\pi'} 
		\ar[d]^{\gamma} & S' 
		\ar[d]^{\tau'} \\
		\Delta^L(S) & \Delta^L(R) \ar[l]^{\Delta^L(\pi)} 
		\ar[r]_{\Delta^L(\pi')}  & \Delta^L(S')
	}\]
	We say that $s\in S$ and $s'\in S'$ are \emph{$\coalg$-bisimilar} if there exists a $\coalg$-bisimulation $R$ such that $s\mathrel{R} s'$. 
	We denote this as $s\sim^{\coalg} s'$.
\end{definition}

It is clear that ${\sim}^\coalg \sbq{\sim}^\catsp$ since a coalgebraic bisimulation is a particular case of a span in which a \textit{relation} has an LMP structure and the projections are zigzag morphisms.
Proposition~\ref{prop:span-implies-ext} showed that ${\sim}^\catsp \sbq{\sim}^\times$, and therefore, we have ${\sim}^\coalg\sbq {\sim}^\catsp \sbq {\sim}^\times$ (\cite[Prop.~5.2.4]{bacci} gives a direct proof of ${\sim}^\coalg \sbq {\sim}^\times$). 

In \cite[Thm.~16]{Bacci2014BisimulationOM}, it is stated that $\sim^\times$ and $\sim^{\coalg}$ are equivalent. 
To show this, the authors argue that every external bisimulation is also
coalgebraic, first describing in Proposition~13 (op.cit.) a technical procedure 
to construct a measure  $\mu\wedge \nu$ on a subset $R$ of product of two
measure spaces $(X,\mu)$ and $(Y,\nu)$ 

As Gburek already observed in 2016 \cite{gburek}, Proposition~13 does
not hold in general, but he proves it in the case of $R$ being a
countably-separated “quasi-equivalence” (a $z$-transitive relation $R\sbq
X\times Y$ that projects surjectively onto $X$ and $Y$)
with $X$ and $Y$ Polish.

Next we offer a counterexample to the Proposition for a
quasi-equivalence $R$ when $Y$ is not Polish.
\begin{example}
	Recall that $I=(0,1)$. Consider the measurable space $X=(I,
        \Borel(I)_V)$ equipped with $\mu\defi \leb_0$, and $Y=(I,\Borel(I)_V)
        \oplus (I, \Borel(I))$ equipped with $\nu\defi \leb_1$	supported on
        $I\oplus \emptyset$.
	Without loss of generality, assume that $\leb_0(V)<\leb_1(V)$.
	Let $R\defi \{(a,\inl(a))\mid a\in I\}\cup\{(a,\inr(a))\mid a\in I\}$ ($R$ is 
	the union of the two ``diagonals'') and let $\pi_X:R\to X$ and $\pi_Y:R\to Y$ be the canonical projections.
	Observe that for every $E$ and $F$ in the respective $\sigma$-algebras, the implication 
        \begin{equation}\label{eq:hip-bacci}
          \pi_X^{-1}[E] =\pi_Y^{-1}[F]  \implies \mu(E)=\nu(F)
        \end{equation}
        holds: If the antecedent is true, the definition of $R$ implies that $F=E\oplus E$; in addition, the second summand of $Y$ forces $E$ to be Borel. 
	\cite[Prop.~13]{Bacci2014BisimulationOM} states that under
        (\ref{eq:hip-bacci}), there should exist a measure $\mu\wedge \nu$ such that
        $(\mu \wedge \nu)(\pi_X^{-1}[V]) = \mu(V)$ and $(\mu \wedge
        \nu)(\pi_Y^{-1}[V\oplus \emptyset]) = \nu(V\oplus \emptyset)$. Hence
        \[
          (\mu \wedge \nu)(\pi_X^{-1}[V])=\leb_0(V)<  
	  \leb_1(V) = (\mu \wedge \nu)(\pi_Y^{-1}[V\oplus \emptyset]).
        \]
	But $\pi_Y^{-1}[V\oplus \emptyset] \subseteq \pi_X^{-1}[V]$, 
	leading to a contradiction. 
\end{example}

In any case, \cite[Thm.~16]{Bacci2014BisimulationOM} does not hold since Example~\ref{exm:ext-not-cat} shows $ {\sim}^\times \nsubseteq {\sim}^\coalg $. 
Nevertheless, the question about the equivalence between the bisimilarities $\sim^\coalg$ and $\sim^\catsp$ remains.

\begin{question}\label{question:span-implies-coalg}
	Does ${\sim}^\catsp \sbq {\sim}^\coalg$ hold? 
	In other words, if two states are bisimilar via a span, does a $\coalg$-bisimulation also exist that relates them?
\end{question}

An antecedent to this question can be found in \cite[Thm.~5.8]{Vink1997BisimulationFP}, where it is proven that for LMP in the category of ultrametric spaces with expansive maps, every $\times$-bisimulation that admits a Borel decomposition is also a $\coalg$-bisimulation. 
However, this condition seems to be too restrictive.

Given a span $(\lmp{T},f,g)$, we could ask about the possibility of constructing a coalgebraic bisimulation that includes $(f\times g)[T]$.
However, this might be too much to ask for, as the question of whether ${\sim}^\catsp \sbq {\sim}^\coalg$ does not necessarily imply that the \textit{same} relation attests $\catsp$-bisimilarity.

With this in mind, let us consider a very simple particular case where the answer to Question~\ref{question:span-implies-coalg} is positive. 
We recall that if $(X,\Sigma_X)$ is a measurable space and $h:X\to Y$, then $\Sigma_Y=\{B\sbq Y\mid h^{-1}[B]\in \Sigma_X\}$ is the largest $\sigma$-algebra on $Y$ that makes $h$ measurable. 
This is called the final $\sigma$-algebra with respect to $h$.

\begin{prop}\label{prop:monic-span-case}
	Let $\lmp{S}\leftarrow \lmp{T} \rightarrow \lmp{S}'$ be a monic span.
	Then the relation $R\defi (f\times g)[T]$ with the final $\sigma$-algebra with respect to $f\times g$, $\Sigma_R$, is an LMP if we define $\tau_a((f(t),g(t)),Q)\defi \tau_a^{\lmp{T}}(t,(f\times g)^{-1}[Q])$.
\end{prop}
\begin{proof}
	Since $\Set$ has products, the injectivity of the span is equivalent to requiring $f\times g:T\to S\times S'$ to be injective. 
	Thanks to this, $\tau_a$ is well-defined. 
	For a fixed $r=(f(t),g(t))$, $\tau_a(r,\cdot)$ is a subprobability measure. 
	If we fix $Q\in \Sigma_R$, 
	\begin{align*}
		\tau_a(\cdot,Q)^{-1}[(0,q)] &= \{(f(t),g(t))\in R\mid \tau_a((f(t),g(t)),Q)<q\} \\
		& = \{(f(t),g(t))\in R\mid \tau_a^{\lmp{T}}(t,(f\times g)^{-1}[Q])<q\} \\
		&= (f\times g)[\{t\in T\mid \tau_a^{\lmp{T}}(t,(f\times g)^{-1}[Q])<q\}].
	\end{align*}
	Given that $f\times g$ is injective, we have $(f\times
        g)^{-1}[(f\times g)[C]]=C$ for any set $C$.  Hence, by the
        definition of $\Sigma_R$, $(f\times g)[\{t\in T\mid
            \tau_a^{\lmp{T}}(t,(f\times g)^{-1}[Q])<q \}]\in \Sigma_R$
        and we conclude that $\tau_a(\cdot, Q)$ is measurable.
\end{proof}

This result is quite weak, as a monic span in $\Set$ is (isomorphic to) a relation. 
Thus, the assumptions practically require a coalgebraic bisimulation. 
However, it does provide some insight when attempting to find a counterexample to Question~\ref{question:span-implies-coalg}: Avoid considering injective spans (and therefore neither of the morphisms can be injective).

In \cite{sokolova05}, properties of coalgebraic bisimulation are listed (with proofs in \cite{Rutten00}). 
The functors used there possess a particular property: They preserve weak pullbacks, or weakly preserve pullbacks. 
The second condition is weaker than the first, but both are equivalent if the category has pullbacks. 
The \textit{pullback} of a cospan between two objects is a span between the same objects completing the diagram to form a commutative square, while also satisfying a universal property. 
The dual concept is that of a \textit{pushout}. 
A weak pullback completes the commutative square without uniqueness in the universal property, and a semi-pullback only completes the commutative square.
The importance of a functor preserving these constructions in some manner lies
in the fact that bisimulations between their coalgebras exhibit good properties
\cite[p.10]{viglizzo05}. At the end of next Section~\ref{sec:catco-bisimulations} we
will resume the discussion of the relationship between categorical and
coalgebraic bisimilarity since it also involves the bisimilarity defined in
terms of cospans.

\subsection{$\catco$-bisimulations}
\label{sec:catco-bisimulations}

We will now address event bisimulations in a categorical context.  In
Corollary~\ref{cor:logic-charac-event}, we stated the logical
characterization of event bisimilarity via the modal logic $\Logic$.
Based on Lemma~\ref{lem:zigzag-logic}, we can conclude that if $f$ and
$g$ form a span of zigzag morphisms, then $f(w)$ and $g(w)$ satisfy
the same formulas of $\Logic$.  Similarly, if $j$ and $k$ now form a
cospan of zigzag morphisms and $j(s) = k(s')$, then $s$ and $s'$ are
logically equivalent.  It is noteworthy that logical equivalence
allows for an external perspective and behaves like an equivalence
relation, whereas event bisimilarity as a sub-$\sigma$-algebra is
internal to the LMP.

Cospans are an appropriate categorical representation of equivalence
relations; in this direction, \cite{coco} indicates that event
bisimilarity corresponds to surjective zigzag cospans in the category
$\LMP$.  An advantage of this approach is that the transitivity of
bisimilarity using cospans follows directly due to the existence of
pushouts in $\LMP$ (which are essentially pushouts in
$\Mes$). The following definition is adapted from \cite[Def.~6.2]{coco}
to our context.

\begin{definition}\label{def:event-bisimulation-ext}
  An \emph{external event bisimulation between $\lmp{S}$ and $\lmp{S'}$} is a
  cospan of surjections to some object $\lmp{T}$ in the category of
  coalgebras.  We will say that $s$ and $s'$ are
  \emph{$\catco$-bisimilar} if there exists a cospan of morphisms
  $(\lmp{T},f,g)$ such that $f(s) = g(s')$.  In this case, we write $s
  \sim^{\catco} s'$.
\end{definition}

\begin{example}\label{exm:catco-not-ext}
  The function $\mathrm{swap}:\lmp{U}\to \lmp{U}_B$ (see
  Example~\ref{exm:zigzag-not-state-domain}) that maps $s\mapsto t$
  and $t\mapsto s$ (and is the identity on the rest) is zigzag.
  Therefore, $\lmp{U} \xrightarrow{\id} \lmp{U}_B
  \xleftarrow{\mathrm{swap}} \lmp{U}$ is a cospan that witnesses
  $s\sim^\catco t$.
\end{example}

\begin{prop}\label{prop:ctaco-implies-event-bisim}
  If $s \sim^\catco s'$, then they are (internally) event bisimilar in
  $\lmp{S}\oplus\lmp{S}'$.
\end{prop}
\begin{proof}
  Given a cospan $(\lmp{T},f,g)$ such that $f(s)=g(s')$, we must prove
  that there exists a stable $\sigma$-algebra $\Lambda$ in
  $\lmp{S}\oplus\lmp{S}'$ such that
  $\inl(s)\mathrel{\Rel(\Lambda)}\inr(s')$.  Let $\Fam=\{f^{-1}[A]\oplus
  g^{-1}[A]\mid A\in \Sigma_T\}$.  Using the fact that $f$ and $g$ are
  zigzag, it can be observed that for any $A\in \Sigma_T$,
  \[
  \{r\in S\oplus S'\mid \tau^\oplus_a(r,f^{-1}[A]\oplus g^{-1}[A])<q\}
    = f^{-1}[B]\oplus g^{-1}[B]
  \] 
  where $B=\{t\in T\mid \tau^{\lmp{T}}_a(t,A)<q\}$. Consequently,
  $\Fam$ is stable and also forms a $\pi$-system.  By
  \cite[Lem.~5.4]{coco}, the $\sigma$-algebra generated by a stable
  $\pi$-system is also stable.  Therefore, $\Lambda\defi \sigma(\Fam)$
  is stable. Finally, we note that
  \[
    \inl(s)\mathrel{\Rel(\Lambda)} \inr(s')
    \Leftrightarrow \inl(s)\mathrel{\Rel(\Fam)} \inr(s') \Leftrightarrow
    \forall A\in\Sigma_T,\,(f(s)\in A\Leftrightarrow g(s')\in A),
  \]
  which is true by definition of $\Fam$.
\end{proof}

Actually, the surjectivity condition required by
Definition~\ref{def:event-bisimulation-ext} was not used in
the previous proof. Furthermore, this
requirement is an obstacle for the comparison between that definition
and logical equivalence. Consider the following finite
LMPs:
\[
  \xymatrix @R=1mm{
    \lmp{S}: &s_1 \ar[ddr]^a&   &s_2 \ar[ddl]_a&&\lmp{S'}:&s'_1 
    \ar[dd]^a&s'_4 \ar@(ur,dr)[]^a \\ 
    &              &   &              &&             &              &                     \\ 
    &              &s_3&              &&             &s'_3          &                     \\
  }
\]

We note that the sets $\{ s_1,s_2,s'_1 \}$, $\{ s_3,s'_3 \}$, and $\{
  s'_4 \}$ are classes of logical equivalence in the sum
$\lmp{S}\oplus \lmp{S}'$ (the algebra generated by these sets is
stable), and zigzag morphisms must preserve this property. Therefore,
if $f:\lmp{S}\to \lmp{T}$ and $g:\lmp{S'}\to \lmp{T}$ are zigzag, then
$g(s'_4)\notin \mathrm{Im}(f)$, and thus $f$ cannot be surjective.

\subsubsection*{Quotients and finality}

The converse to Proposition~\ref{prop:ctaco-implies-event-bisim} also
holds, but in order to prove it we need to construct a “quotient”
of sorts; we will state most of the auxiliary results
about it without proof.

For an LMP $\lmp{S}=(S,\Sigma,\{\tau_a \}_a)$ and  $\Lambda$ a stable sub-$\sigma$-algebra, consider the
projection $\pi: S \to S/\Rel(\Lambda)$. We can endow
$S/\Rel(\Lambda)$ with the following $\sigma$-algebra and kernels 
\[
  \Lambda/\mathcal{R}(\Lambda)\defi \{\pi[Q] 
    \mid Q\in \Lambda \},
  \qquad
  \bar{\tau}_a(\pi(s),Q')\defi\tau_a(s,\pi^{-1}[Q']),
\]
where $s\in S$ and $Q' \in \Lambda/\mathcal{R}(\Lambda)$. These kernels
are well defined and then
\[
  \lmp{S}/\Lambda\defi
  (S/\mathcal{R}(\Lambda),\Lambda/\mathcal{R}(\Lambda),\{\bar{\tau}_a
  \}_a)
\]
is an LMP and the canonical projection $\pi:\lmp{S}\to
\lmp{S}/\Lambda$ is upgraded to a zigzag morphism. An antecedent of this construction is
\cite{DBLP:journals/siamcomp/Doberkat05}, where the whole setting is
restricted to analytic state spaces and countably generated $\Lambda$.

It is to be noted
that this is not a quotient “over $\Set$”, but it will
be shown to be one “over $\Mes$” since it is an instance of the following particular
case of \cite[Def.~2.9]{dubuc2006topological}:
\begin{definition}
  Let $\functor{U} : \mathbf{T} \to \mathbf{S}$ be a functor, $\lmp{X}_0,\lmp{X}\in
  \mathbf{T}$ and $f_0: \lmp{X}_0 \to \lmp{X}$.  The morphism $f_0$ is \emph{$\functor{U}$-final}
  if for any $g_0 : \lmp{X}_0 \to \lmp{Y}$ in $\mathbf{T}$ and arrow $\phi:
  \functor{U}(\lmp{X}) \to \functor{U}(\lmp{Y})$ in $\mathbf{S}$ such that $\functor{U}(g_0)=\phi\circ
  \functor{U}(f_0)$, there exists a unique $g : \lmp{X} \to \lmp{Y}$ such that
  $\functor{U}(g)=\phi$ and $g \circ f_0 = g_0$.
\end{definition}

For the rest of this section, let $\functor{G}$ and $\functor{V}$  be the forgetful functors
from $\LMP$ to $\Mes$  and from $\LMP$ to $\Set$, respectively. Next, we
will briefly discuss $\functor{G}$- and $\functor{V}$-finality of
zigzag morphisms.

\begin{lemma}\label{lem:pi-is-G-final}
  Every surjective zigzag $\pi:\lmp{S}\to \lmp{S}'$ is $\functor{G}$-final.
\end{lemma}
\begin{proof}
  Let $g_0:\lmp{S}\to \lmp{T}$ be a zigzag morphism and
  $\phi:\functor{G}(\lmp{S}') \to \functor{G}(\lmp{T})$ a measurable map such
  that $\functor{G}(g_0)=\phi\circ \functor{G}(\pi)$. We complete the diagram in
  Figure~\ref{fig:final} by defining $g:\lmp{S}' \to
  \lmp{T}$ by $g(\pi(s))=g_0(s)$.  To check well-definedness, assume
  that $\pi(s)=\pi(s')$, which is the same as
  $\functor{G}(\pi)(s)=\functor{G}(\pi)(s')$. By applying $\phi$, we have 
  $\functor{G}(g_0)(s)=\functor{G}(g_0)(s')$ by hypothesis, which is again the same as
  $g_0(s)=g_0(s')$.
  As $g$ and $\phi$ coincide as maps between
  sets, $g$ is measurable.  Let $B\sbq T$ be measurable; then
  \begin{multline*}
    \tau_a^\lmp{T}(g(\pi(s)),B)= \tau_a^\lmp{T}(g_0(s),B) =
    \tau_a^\lmp{S}(s,g_0^{-1}[B]) =\\
    = \tau_a^\lmp{S}(s,\pi^{-1}[g^{-1}[B]]) =
    \tau_a^{\lmp{S}'}(\pi(s),g^{-1}[B]).
  \end{multline*}
  This proves that $g$ is zigzag. As $\functor{G}$ is faithful, $g$ is
  unique.
\end{proof}
\begin{figure}
  \begin{center}
    \begin{tikzpicture}
      \matrix (m) [ matrix of math nodes, row sep=3em,
        column sep=3em, nodes in empty cells,
        nodes={anchor=center} ] { \functor{G}(\lmp{S}) & & & &
        \functor{G}(\lmp{S}') \\[-3em] & \lmp{S} & &
        \lmp{S}' & \\[1em] & & & \lmp{T} &\\[-3em] &
        & & & \functor{G}(\lmp{T}) \\ };
      
      \draw[->] (m-1-1) -- node[above]{$\functor{G}(\pi)$} (m-1-5);
      \draw[->, bend right=28] (m-1-1) to
      node[left]{$\functor{G}(g_0)\quad$} (m-4-5); \draw[->] (m-2-2) --
      node[above]{$\pi$} (m-2-4); \draw[->] (m-2-2) to
      node[left]{$\forall g_0$} (m-3-4); \draw[->, dashed]
      (m-2-4) -- node[right]{$\exists !g$} (m-3-4);
      \draw[->] (m-1-5) --node[right]{$\forall \phi$}
      (m-4-5);
    \end{tikzpicture}
  \end{center}
 \caption{The projection $\pi$ is $(\LMP\to\Mes)$-final.\label{fig:final}}
\end{figure}
This result slightly strengthens the idea proposed in the last
paragraph \cite[Sect.~6]{coco}, were epi-mono factorizations of zigzag
morphisms are discussed. As it is expected, not all epimorphims behave
as in other more structured categories:
\begin{example}
  The map $\pi:\lmp{S}\to \lmp{S}/\Lambda$ may not be $\functor{V}$-final.  Take $\lmp{S}\defi \lmp{U}$,
  $\Lambda \defi \Borel(\Inter)\oplus \Power(\{ s,t,x \})$, the zigzag
  $\id:\lmp{U}\to \lmp{U}$, and the function $\phi:U/\Rel(\Lambda) \to
  U$ given by $\phi(\pi(s))=s$ (this is well defined since
  $\Rel(\Lambda)$ is the identity relation).  Then there is no
  measurable $g:\lmp{U}/\Lambda \to \lmp{U}$ such that $\id=g\circ
  \pi$ and $\functor{V}(g)=\phi$.

  Since this  $\lmp{U}/\Lambda$ is isomorphic to $\lmp{U}_B$ from
  Example~\ref{exm:zigzag-not-state-domain}, we have obtained that the
  identity morphism $\lmp{U}\to \lmp{U}_B$ is $\functor{G}$-final but not
  $\functor{V}$-final.
\end{example}

We will need the following basic result for studying $\functor{V}$-finality:
\begin{lemma}[{\cite[Lem.~4.5]{coco}}]\label{lem:preimage-stable}
  If $\pi: \lmp{S} \to \lmp{S}'$ is zigzag morphism, then
  $\{ \pi^{-1}[Q] \mid Q \text{ measurable in
    }\lmp{S}' \}$ is a stable sub-$\sigma$-algebra of $\Sigma$.
\end{lemma}

\begin{theorem}\label{th:charact-V-final}
  Let $\pi: \lmp{S} \to \lmp{S}'$ be a surjective zigzag morphism, $R\defi
  \{ (s,t) \mid  \pi(s) = \pi(t)\}$ be its \emph{kernel}, and $\Lambda
  \defi \{ \pi^{-1}[Q]  \mid Q \text{ measurable in }\lmp{S}'
  \}$.
  The following are equivalent:
  \begin{enumerate}
  \item\label{item:pi-V-final} $\pi$ is $\functor{V}$-final.
  \item\label{item:greatest-stable} $\Lambda$ is the greatest
    stable sub-$\sigma$-algebra of $\Sigma(R)$.
  \end{enumerate}
\end{theorem}

\begin{proof}
  For \ref{item:pi-V-final}$\Rightarrow$\ref{item:greatest-stable},
  first note that $\Lambda$ consists of $R$-closed subsets by
  construction and it is stable by Lemma~\ref{lem:preimage-stable}.

  Assume $\Lambda_*$ is a stable sub-$\sigma$-algebra of
  $\Sigma(R)$. It follows that
  $R\sbq\Rel(\Lambda_*)$ and there is a canonical
  map $\phi : \functor{V}(\lmp{S}') \to
  \functor{V}(\lmp{S}/\Lambda_*)$ satisfying $\phi \circ
  \functor{V}(\pi) = \functor{V}(\pi_*)$, where $\pi_* : \lmp{S}
  \to \lmp{S}/\Lambda_*$ is  the projection zigzag. By considering $g_0$ in
  Figure~\ref{fig:final} to be $\pi_*$, there must exist a zigzag $g:\lmp{S}' \to
  \lmp{S}/\Lambda_*$ with $\functor{V}(g) =\phi$ making the diagram commute. We show that every
  $A\in \Lambda_*$ belongs to $\Lambda$.
  Note that $A = \pi_*^{-1}[\pi_*[A]] =\pi^{-1}[g^{-1}[\pi_*[A]]]$;
  since $A\in \Lambda_*$, $\pi_*[A]$ is measurable and one can take
  $Q\defi g^{-1}[\pi_*[A]]$ in the definition of $\Lambda$.
  
  For \ref{item:greatest-stable}$\Rightarrow$\ref{item:pi-V-final}, let $g_0:\lmp{S}\to \lmp{T}$ be a zigzag morphism and $\phi:\functor{V}(\lmp{S}') \to \functor{V}(\lmp{T})$ a map such
  that $\functor{V}(g_0)=\phi\circ \functor{V}(\pi)$.
  Analogous to the proof of Lema~\ref{lem:pi-is-G-final}, let
  $g:\lmp{S}' \to \lmp{T}$ be given by $g(\pi(s))=g_0(s)$. 
  By Lemma~\ref{lem:preimage-stable}, $\Theta \defi \{ g_0^{-1}[C]  \mid C \text{ measurable in }\lmp{T} \}$
  is stable. Since $  g_0^{-1}[C] = \pi^{-1}[g^{-1}[C]]$, $\Theta \sbq \Sigma(R)$ and hence
  $\Theta \sbq \Lambda$ by maximality.

  We only have to check that $g$ is measurable. Take $B\sbq T$
  measurable; it is enough to show that $g^{-1}[B] = \pi[A]$ for
  some $A\in\Lambda$. We choose  $A\defi g_0^{-1}[B] \in \Theta \sbq
  \Lambda$. Hence $\pi[A] = \pi[g_0^{-1}[B]] =
  \pi[\pi^{-1}[g^{-1}[B]]] = g^{-1}[B]$, by surjectivity of $\pi$.
\end{proof}
We immediately obtain:
\begin{corollary}\label{cor:charact-V-final}
  Let $\Lambda\sbq \Sigma$ be a stable $\sigma$-algebra. The
  following are equivalent:
  \begin{enumerate}
  \item $\pi:\lmp{S} \to \lmp{S}/\Lambda$ is $\functor{V}$-final.
  \item $\Lambda$ is the greatest
    stable sub-$\sigma$-algebra of $\Sigma(\Rel(\Lambda))$.\qed
  \end{enumerate}
\end{corollary}

In particular,  if we alter the domain of the projection and view it as a function $\pi:(S,\Lambda,\{\tau_a\}_a) \to \lmp{S}/\Lambda$, we will always get a $\functor{V}$-final morphism.
A sufficient condition that appears in many contexts for $\functor{V}$-finality (i.e., it implies the second
item of Corollary~\ref{cor:charact-V-final}) is $\Lambda=\Sigma(\Rel(\Lambda))$;
we do not know if it is necessary.

Back to our discussion of bisimilarities, we can show:
\begin{theorem}\label{th:catco_eq_event}
  The relation $\sim^\catco$ coincides with the restriction of event
  bisimilarity on $\lmp{S}\oplus\lmp{S}'$ to $\inl[S] \times
  \inr[S']$.
\end{theorem}
\begin{proof}
  The ($\Rightarrow$) direction is
  Proposition~\ref{prop:ctaco-implies-event-bisim}. For
  ($\Leftarrow$), assume $\inl(s)$ and $\inr(s')$ are event bisimilar in
  $\lmp{S}\oplus\lmp{S}'$. Let $\Lambda$ be a stable $\sigma$-algebra
  on $\lmp{S}\oplus\lmp{S}'$ such that $\inl(s)\mathrel{\Rel(\Lambda)}\inr(s')$. Then
  \begin{equation*}%
    \lmp{S} \xrightarrow{\pi\circ\, \inl}  
    (\lmp{S}\oplus\lmp{S}')/\Lambda \xleftarrow{\pi\circ\, 
      \inr} \lmp{S}'
  \end{equation*}
  is a cospan of zigzags witnessing $s\sim^\catco s'$.
\end{proof}
Validity sets of formulas are preserved under sum, so we can conclude
that
\begin{corollary}\label{cor:catco-implies-logic}
  External event bisimilarity coincides with logical equivalence.
\end{corollary}

Since internal bisimilarity implies event bisimilarity, we then have:
\begin{corollary}\label{cor:state-impl-catco}
  $s \mathrel{(\sim_\sfs)}_\times s'$ implies $s\sim^\catco s'$. \qed
\end{corollary}

\subsubsection*{From $\catco$ to $\catsp$ and coalgebraic bisimilarity}

From \cite{coco}, we know that every span of LMP can be completed to a
commutative square, and hence ${\sim^\catsp} \subseteq {\sim^\catco}$ holds in
general. Whenever cospans can be completed to square, we have the reverse inclusion.
\begin{lemma}\label{lem:semi-pullbacks}
  For every category of LMP having semi-pullbacks, ${\sim^\catsp} =
  {\sim^\catco}$.
\end{lemma}
The hypothesis does not hold in general as it was shown in \cite[Thm.~12]{Pedro20111048}. 
Nevertheless, a fair amout of research has been directed to the recovery of semi-pullbacks
by considering LMP over better behaved spaces.
An important antecedent is \cite{Edalat}, which deals with LMP over analytic spaces with universally measurable transitions. 
Recall that \cite{Desharnais} introduced \textit{generalized} LMP to take advantage of this result. 
These are analogous to LMP with an analytic state space but the functions $\tau_a(\cdot,A)$
are only required to be universally measurable for each $A\in \Sigma$.

Later it was proved that the category of LMP over analytic spaces has
semi-pullbacks (more precisely, for the category of stochastic relations over
analytic spaces), and furthermore, the upper vertex of the required span can be
chosen to be Polish \cite[Thm.~5.2]{Doberkat:2005:SSR:1089905.1089907}.

The most recent results in this direction ensure the existence of semi-pullbacks
for LMP over separable metrizable spaces where $\tau_a(s,\cdot)$ is a Radon measure
for each $s$ and $a$ \cite[Thm.~5.1]{PachlST} and, in particular, for LMP over
universally measurable separable spaces \cite[Thm.~5.3]{PachlST}.

It can be checked from the proofs in \cite{PachlST} that the semi-pullbacks
constructed in each case turn
out to be pullbacks in $\Set$ and therefore they are relations. We conclude:
\begin{theorem}\label{th:coalg-cat-coanalytic}
  Coalgebraic bisimilarity $\sim^\coalg$, external event bisimilarity
  $\sim^\catco$, and hence also categorical bisimilarity $\sim^\catsp$ are all the same on the categories of
  LMP over:
  \begin{enumerate}
  \item
    universally measurable separable spaces;
  \item coanalytic sets; and
  \item {\cite{Doberkat:2005:SSR:1089905.1089907}} analytic sets.
  \end{enumerate}
\end{theorem}

\subsection{External bisimilarities based on sums}
\label{sec:oplus-bisimilarity}

Now in the general setting of two possibly different LMP $\lmp{S}$ and $\lmp{S}'$, one possibility is the approach taken in \cite[Def.~3.5.1]{Desharnais}.
We adapt the name of the bisimulation defined there to our purposes.

\begin{definition}\label{def:oplus-bisimulation}
	An \emph{$\oplus$-bisimulation} (read ``oplus bisimulation'') between $\lmp{S}$ and $\lmp{S'}$ is an equivalence relation $R$ on $S\oplus S'$ such that for $s \in S$ and $s' \in S'$ with $\inl(s)\mathrel{R}\inr(s')$, and for every $A\oplus A'\in (\Sigma\oplus\Sigma')(R)$, the following holds:
	\[ \forall a\in L \; \tau_a(s,A)=\tau_a'(s',A').\]
\end{definition}

We mentioned that the interest in defining an external state bisimulation $R$ is to detect two bisimilar states, one in each process. 
The above definition takes this premise to the extreme: It does not impose any conditions on the Markov kernel values for pairs of states in the \textit{same} process; the only requirement is that the relation is an equivalence. 
This is unfortunate because, for instance, such a relation restricted
to the first space could identify  points that are not $R$-related
to any other in the second space. In an extreme case, $R$ could be the union of two equivalence relations, one in each process.

Based on the definition of $\tau^\oplus$ and the equivalence \eqref{eq:equiv-oplus-kernels}, we can see that the condition for being an $\oplus$-bisimulation is the same as that of internal bisimulations in $\lmp{S}\oplus\lmp{S}'$, except that it applies to fewer pairs. 
From this, we immediately deduce the following straightforward result.

\begin{lemma}\label{lem:internal-equiv-implies-oplus}
	If $R$ is an internal bisimulation on the sum $\lmp{S}\oplus \lmp{S'}$ and is an equivalence relation, then it is an $\oplus$-bisimulation between $\lmp{S}$ and $\lmp{S'}$.
\end{lemma}

In particular, this implies that  state bisimilarity on $\lmp{S}\oplus \lmp{S'}$ is an $\oplus$-bisimulation. 
On the other hand, considering the observation following Definition~\ref{def:oplus-bisimulation}, it is clear that not every $\oplus$-bisimulation $R$ is an internal bisimulation on $\lmp{S}\oplus \lmp{S'}$. 
For example, the relation $R=\id_{U\oplus U}\cup \{(\inl(s),\inl(t)),(\inl(t),\inl(s))\}$ is an $\oplus$-bisimulation between $\lmp{U}$ and $\lmp{U}$, but this relation is not an internal bisimulation in $\lmp{U}\oplus\lmp{U}$ because $V\oplus \emptyset$ is $R$-closed. 
Furthermore, it cannot be contained in any internal bisimulation on the sum, as shown in Proposition~\ref{prop:state-bisim-in-sum}.

We complete Definition~\ref{def:oplus-bisimulation} with the corresponding bisimilarity.
Instead of defining it as a relation on the sum, to facilitate comparisons with other bisimilarity relations, we will define it as follows:

\begin{definition}\label{def:sim-oplus}
	If $s\in S$ and $s'\in S'$, we will say that $s\sim^\oplus s'$ if
	there exists an $\oplus$-bisimulation $R$ between $\lmp{S}$ and
	$\lmp{S'}$ such that $(\inl(s),\inr(s'))\in R$.
\end{definition}

This way, ${\sim^\oplus}\sbq S\times S'$, and therefore it cannot be an $\oplus$-bisimulation. 
Note that in the definition, we state that two states can be $\oplus$-bisimilar only if they are in different summands.

It is worth noting that if $R$ is an $\oplus$-bisimulation, then its
descent~\eqref{eq:R_times} $R_\times$ is a subset of ${\sim}^\oplus$. 
If we write $T$ to be the union of all $\oplus$-bisimulations (which is the usual way to define bisimilarity), then it follows that $T_\times={\sim}^\oplus$.

From this, we can conclude a simple consequence of the previous Lemma~\ref{lem:internal-equiv-implies-oplus}:
\begin{corollary}\label{cor:R-oplus-sbq-sim-oplus}
	$({\sim}_\sfs)_\times \sbq {\sim}^\oplus$.
\end{corollary}

The reverse inclusion does not hold. 
Following Lemma~\ref{lem:internal-equiv-implies-oplus}, we have given an example showing that not every $\oplus$-bisimulation is an internal bisimulation on $\lmp{S}\oplus \lmp{S'}$.
We will see in the next example that we can also have an $\oplus$-bisimulation $R$ between $\lmp{S}$ and $\lmp{S'}$ such that $R_\times$ cannot be included in any $\times$-bisimulation.
In particular, $R\cap (\inl[S]\times\inr[S'])$ does not lie within any internal bisimulation on the sum.

\begin{example}\label{exm:oplus-notsbq-cospan}
	Consider $\lmp{U}_B$ from Example~\ref{exm:zigzag-not-state-domain} and versions of $(\lmp{U}_B)_s$ and $(\lmp{U}_B)_t$	(cfr. Example~\ref{exm:ext-bisim-not-transitive}) where the $\infty$-label kernels have been replaced by \begin{equation*}
		\begin{aligned}
			\tau^{(\lmp{U}_B)_s}_\infty(r,A)&\defi \chi_{\{s \}}(r)\cdot \leb(A), \\
			\tau^{(\lmp{U}_B)_t}_\infty(r,A)&\defi \chi_{\{t \}}(r)\cdot \bigl( \tfrac{1}{2} \cdot \delta_{\frac{1}{4}}(A\cap (0,\tfrac{1}{2})) + \leb(A\cap [\tfrac{1}{2},1))\bigr).
		\end{aligned}
	\end{equation*}
	First, let us show that $s$ and $t$ are not logically equivalent. 
	Note that $\forall q\in\Inter\cap\Q \, \sem{\pos{q}_{\geq 1}\top} = B_q$.
	For any $q\in \Inter\cap\Q$ such that $q<\frac{1}{4}$, a state $r$ satisfies $\pos{\infty}_{\geq q}\pos{q}_{\geq 1}\top$ if and only if $\tau_\infty(r,B_q)\geq q$.
	Given that $\tau^{(\lmp{U}_B)_s}_\infty(s,B_q)=q$ and $\tau^{(\lmp{U}_B)_t}_\infty(t,B_q)=0$, we can distinguish $s$ and $t$ with a formula.

	Now, take any non Borel $D \sbq (\frac{1}{2},1)$.
	Let $R$ be the equivalence relation on the sum $(\lmp{U}_B)_s \oplus
        (\lmp{U}_B)_t$ whose equivalence classes are $\{\inl(s),\inr(t)\}$,
        $\{ \inl(r),\inr(r)\}$ with $r\in D\cup\{x\}$,
        $\{\inl(r)\mid r\in I\sm D \}$,
        and $\{\inr(r)\mid r\in I\sm D\}$.
	That is, $R$ identifies $s$ and $t$, each point in $D\cup\{x\}$ with its copy on the other process, and clusters each copy of the set $I\sm D$ independently.
	We will show that $R$ is an $\oplus$-bisimulation.
	
	Let $A$ be an $R$-closed measurable set. Proving that the
	$\oplus$-bisimulation definition holds for pairs $(\inl(r),\inr(r))$
	is direct, as $\inl(x)\in A \iff \inr(x)\in A$ is true.
	
	For $(\inl(s),\inr(t))$, we need to show that $\tau_\infty^\oplus(\inl(s),A) = \tau_\infty^\oplus(\inr(t),A)$.
	If $A$ has empty intersection with $I\sm D\oplus I\sm
	D$, it is of the form $B\oplus B$ with $B$ Borel measurable and
	the conclusion holds.
	
	Otherwise, and without loss of generality, $\inl(r)\in A$ for some
	$r\in I\sm D$. Then $\inl[I\sm D]\sbq A$. Note also that
	$\inr[(A\cap \inl[D])\rest (\lmp{U}_B)_s] = A \cap \inr[D]$.
	Moreover, given that $A\cap \inl[I]$ is Borel, it
	must be the case that $A\cap \inl[D]$ is not in $\Borel(I)$, otherwise $D$ would be Borel.
	Therefore $A\cap \inr[D] \subsetneq A \cap \inr[I]$ and finally $\inr[I\sm D]\sbq A$.
	
	Hence, $A\cap (I\oplus I)$ can only be of the form $B\oplus B$ where $B$ is a
	Borel subset of $I$ satisfying $B\subset D$ or $I\sm D \subset B$.
	As a consequence, the required measures coincide.
	 
	Note that all cases have been considered,
        as an $\oplus$-bisimulation only requires checking states in different summands.
	Thus, the relation $R$ witnesses that $s\sim^\oplus t$. 
	In summary, ${\sim}^\oplus \nsubseteq {\sim}^\catco$.
	
\end{example}

Using this Example and Corollary~\ref{cor:state-impl-catco}  we conclude that ${\sim}^\oplus \nsubseteq ({\sim}_\sfs)_\times$.
Moreover, Corollary~\ref{cor:ext-bis-sbq-R-oplus} then implies that ${\sim}^\oplus \nsubseteq {\sim}^\times$.

Next, we will explore a connection between $\oplus$- and $\times$-bisimulations.
Corollaries~\ref{cor:R-oplus-sbq-sim-oplus} and \ref{cor:ext-bis-sbq-R-oplus} indicate that ${\sim}^\times\sbq ({\sim}_\sfs)_\times \sbq {\sim}^\oplus$. 
We observe that in Example~\ref{exm:oplus-notsbq-cospan}, the equivalence relation there is not measurable. 
In the following, we will give sufficient measurability conditions for $R_\times$ to be an $\times$-bisimulation whenever $R$ is an $\oplus$-bisimulation.

The following concept will be useful: A pair $(A,A')$ is \emph{saturated} if $R[A]=A'$ and $A=R^{-1}[A']$.

\begin{lemma}\label{lem:R-oplus-aux}
	Let $R$ be an $\oplus$-bisimulation between $\lmp{S}$ and $\lmp{S'}$.
	\begin{enumerate}
		\item\label{item:R-oplus-aux-sets} If $(A,A')$ is a $R_\times$-saturated pair, then $A\oplus A'$ is $R$-closed.
		\item\label{item:R-oplus-aux-kernels} If $\dom(R_\times)\in \Sigma$, then $\forall s\in \dom(R_\times)\; \forall a\in L \; \tau_a(s,S\sm \dom(R_\times))=0$.
	\end{enumerate}
\end{lemma}
\begin{proof}
  For the first part, let $x\in A$ and consider the case $\inl(x)\mathrel{R}\inr(y)$.
  Since $R_\times[A]\sbq A'$, we have $y\in A'$. 
  Now suppose $\inl(x)\mathrel{R}\inl(y)$. 
  Since $(A,A')$ is saturated, there exists $x'\in A'$ such that $x\mathrel{R_\times} x'$. 
  Hence, $\inl(x)\mathrel{R} \inr(x')$.  
  Since $R$ is an equivalence relation, we have $\inl(y)\mathrel{R} \inr(x')$, and as $(R_\times)^{-1}[A']\sbq A$, we conclude that $y\in A$.
  The case for $x$ such that $\inr(x)\mathrel{R}z$ for some $z$ is similar.
		
  For the second item, we first prove that $S_{\mathrm{anti}}\defi (S\sm \dom(R_\times))\oplus \emptyset$ is $R$-closed: Let $\inl(x)\in S_{\mathrm{anti}}$ and assume that $\inl(x)\mathrel{R} z$ for some $z$. 
  Then, there exists $y$ such that $z=\inl(y)$. Otherwise, we would immediately have $x\in \dom(R_\times)$.  
  Now, by way of contradiction, assume $y\in \dom(R_\times)$.  
  Then there exists $y'\in \img(R_\times)$ such that $y\mathrel{R_\times} y'$.
  Consequently, $\inl(y)\mathrel{R} \inr(y')$, and by transitivity of $R$, we have $\inl(x)\mathrel{R}\inr(y')$.  
  Thus, $x\in \dom(R_\times)$, which is a contradiction.
  
  Now, consider $s'\in \img(R_\times)$ such that $s\mathrel{R_\times} s'$.
  Since $S_{\mathrm{anti}}$ is $R$-closed and measurable, for all $a\in L$, we conclude
    \[
      \tau_a(s,S\sm \dom(R_\times)) =
      \tau_a^\oplus(\inl(s),S_{\mathrm{anti}}) =
      \tau_a^\oplus(\inr(s'),S_{\mathrm{anti}}) =
      \tau'_a(s',\emptyset) = 0. \qedhere
    \]
\end{proof}

\begin{prop}\label{prop:case-R-oplus-is-ext}
	Let $R$ be an $\oplus$-bisimulation between $\lmp{S}$ and $\lmp{S'}$ such that $\dom(R_\times)\in \Sigma$ and $\img(R_\times)\in \Sigma'$, then $R_\times$ is a $\times$-bisimulation.
\end{prop}

\begin{proof}
	First, we note that $(\dom(R_\times),\img(R_\times))$ is a $R_\times$-saturated pair.
	If $(s,s')\in R_\times$, by Lemma~\ref{lem:R-oplus-aux}(\ref{item:R-oplus-aux-sets}), we have $\tau_a(s,\dom(R_\times)) = \tau'_a(s',\img(R_\times))$. 
	We use this equality along with the fact that $S\oplus S'$ is $R$-closed and Lemma~\ref{lem:R-oplus-aux}(\ref{item:R-oplus-aux-kernels}) to calculate $\tau'_a(s',S'\sm \img(R_\times))= \tau'_a(s',S')-\tau'_a(s',\img(R_\times) = \tau_a(s,S)-\tau_a(s,\dom(R_\times)) = \tau_a(s,S\sm \dom(R_\times))=0$.
	
	Now, let $(A,A')$ be a $R_\times$-closed measurable pair. We write 
	\begin{align*}
		A &= (A\cap \dom(R_\times)) \cup (A\sm \dom(R_\times)) = \mathrel{\mathop:} A_{\mathrm{sat}}\cup A_{\mathrm{anti}},\\
		A' &= (A'\cap \img(R_\times)) \cup (A'\sm \img(R_\times)) = \mathrel{\mathop:} A'_{\mathrm{sat}}\cup A'_{\mathrm{anti}}.
	\end{align*}
	We observe that $(A_{\mathrm{sat}},A'_{\mathrm{sat}})$ is a $R_\times$-saturated measurable pair, $A_{\mathrm{anti}}\sbq S\sm\dom(R_\times)$, and $A'_{\mathrm{anti}}\sbq S'\sm \img(R_\times)$. 
	Therefore, $\tau_a(s,A_{\mathrm{sat}})= \tau'_a(s',A'_{\mathrm{sat}})$; $\tau_a(s,A_{\mathrm{anti}})\leq \tau_a(s,S\sm \dom(R_\times))=0$, and $\tau'_a(s',A'_{\mathrm{anti}})\leq \tau'_a(s',S'\sm \img(R_\times))=0$. Hence,
	\begin{align*}
	\tau_a(s,A) &= \tau_a(s,A_{\mathrm{sat}})+ \tau_a(s,A_{\mathrm{anti}}) = \tau_a(s',A'_{\mathrm{sat}}) + 0 \\
	&= \tau'_a(s',A'_{\mathrm{sat}}) + \tau'_a(s',A'_{\mathrm{anti}}) = \tau_a(s',A'). \qedhere
	\end{align*}
\end{proof}
\begin{corollary}\label{cor:R-oplus-ext} 
	If $\dom(({\sim}_\sfs)_\times)\in \Sigma$ and $\img(({\sim}_\sfs)_\times)\in \Sigma'$, then $({\sim}_\sfs)_\times = {\sim^\times}$.
\end{corollary}

Proposition~\ref{prop:case-R-oplus-is-ext} also gives some insight about logical equivalence versus $\oplus$-bisimulation.
Notice that in $\lmp{U}_s\oplus \lmp{U}_t$, the logical equivalence cannot be an $\oplus$-bisimulation, since $s$ and $t$ are externally bisimilar.
In spite of this, the following question remains.
\begin{question}\label{question:oplus-cospan-unknown}
	If two states satisfy the same formulas, are they related by some $\oplus$-bisimulation?
\end{question}

Another discussion that arises with the $\oplus$-bisimulation approach is that in the case
of $\lmp{S}=\lmp{S}'$, we would like to reconcile this definition with that of internal state bisimulations.  
However, we cannot make a direct comparison since by definition $\oplus$-bisimulations are equivalences on
$S\oplus S$, and hence subsets of $(S\oplus S)\times (S\oplus S)$; in contrast, internal bisimulations are relations on $S$, that is subsets of $S\times S$.  
We highlight this difference by stating that these
concepts have different “type.”

We can make a statement in one direction: By
Lemma~\ref{lem:state-bisimulation-in-sum}(\ref{item:state-bisimulation-in-sum}), we know that if $R\sbq
S\times S$ is an internal state bisimulation and an equivalence
relation, then $R^+$ is an internal bisimulation on
$\lmp{S}\oplus\lmp{S}$ and an equivalence relation.  
By Lemma~\ref{lem:internal-equiv-implies-oplus}, $R^+$ is also a
$\oplus$-bisimulation between $\lmp{S}$ and $\lmp{S}$.
In the case of ${\sim}_\sfs$, Proposition~\ref{prop:state-bisim-in-sum} even tells us that ${\sim}_\sfs = (({\sim_\sfs})^+)_\times = ({\sim}_{\sfs,\lmp{S}\oplus \lmp{S}})_\times =(\sim_\sfs)_\times \sbq
{\sim}^\oplus$.
\subsubsection*{From $\oplus$-bisimilarity to events}

In \cite[Def.~7.2]{doi:10.1142/p595}, a different definition of $\oplus$-bisimulation is presented. 
In this variation, the equivalence relation is not required to be within the sum $S \oplus S'$, but rather within a possibly larger process of which $\lmp{S}$ and $\lmp{S}'$ are direct summands.
To provide clarity, we have rephrased this definition to explicitly specify the necessary inclusions.

\begin{definition}\label{def:prakash-bisim}
	Let $\lmp{S}$ and $\lmp{S}'$ be a pair of LMPs and $s\in S, s'\in S'$. We say that $s$ and $s'$ are \emph{$\oplus_\mathrm{P}$-bisimilar} if there is a process $\lmp{W}$ and a state bisimulation equivalence on $(\lmp{S}\oplus\lmp{S}')\oplus \lmp{W}$ relating $\mathrm{inl}(\mathrm{inl}(s))$ and $\mathrm{inl}(\mathrm{inr}(s'))$.
\end{definition}

In what follows we will show that this notion is the same as external event bisimilarity, that is, $\sim^\catco$.

\begin{lemma}\label{lem:event-lmp}
	Let $\lmp{T}=(T,\Sigma,\{\tau_a\mid a\in L\})$ be an LMP, and consider the LMP $\lmp{T}_\e\defi (T,\sigma(\sem{\Logic}),\{\tau_a\mid a\in L\})$. 
	Then ${\sim_{\e,\lmp{T}}} = {\sim_{\e,\lmp{T}_\e}} = {\sim_{\s,\lmp{T}_\e}}$.
\end{lemma}
\begin{proof}
	For the first equality we use the zigzag morphism $\mathrm{id}:\lmp{T}\to\lmp{T}_\e$ to obtain $\sem{\Logic}_\lmp{T}= \mathrm{id}^{-1}[\sem{\Logic}_{\lmp{T}_\e}] = \sem{\Logic}_{\lmp{T}_\e}$.
	Thus, we have the identities ${\sim_{\e,\lmp{T}}} = \Rel(\sigma(\sem{\Logic}_\lmp{T})) = \Rel(\sigma(\sem{\Logic}_{\lmp{T}_\e}))= {\sim_{\e,\lmp{T}_\e}}$. 
	
	For the last equality, we only have to show that ${\sim_{\e,\lmp{T}_\e}}$ is a state bisimulation in $\lmp{T}_\e$.
	This is immediate from the equality $\sigma(\sem{\Logic}_{\lmp{T}_\e})=\sigma(\sem{\Logic}_{\lmp{T}_\e})(\sim_{\e,\lmp{T}_\e})$.
\end{proof}

The second equality can also be seen immediately since the Zhou filtration is monotone increasing from $\sigma(\sem{\Logic})$ to $\Sigma$ \cite[Fig.~2]{moroni2020zhou}.
The same analysis shows that there are other $\sigma$-algebras with this property.

\begin{lemma}\label{lem:sum-preserves-bisimil}
	Let $\lmp{T}, \lmp{T}'$ be LMPs and $x,y\in T'$. 
	If $x\sim_{*,\lmp{T}'} y$, then $\mathrm{inr}(x)\sim_{*,\lmp{T}\oplus \lmp{T}'} \mathrm{inr}(y)$ for $*\in \{\sfs,\e\}$.
\end{lemma}
\begin{proof}
	The inclusion of $\lmp{T}'$ into $\lmp{T}\oplus\lmp{T}'$ is a zigzag. It follows that for all $z\in T'$, $z$ and $\mathrm{inr}(z)$ are logically equivalent.
	Thus, the result for $\sim_\e$ follows from the logical characterization of event bisimilarity and the transitivity of the logical equivalence.
	
	In the case of state bisimilarity, we recall that the lift $R_r$ (\ref{eq:lift}) of a state bisimulation $R$ in $\lmp{T}'$ is a state bisimulation on $\lmp{T}\oplus\lmp{T}'$.
\end{proof}

Notice that in the case of event bisimilarity, the converse implication is also true.

\begin{corollary}\label{coro:event-bisim-implies-state-in-sum-event}
	If $i \sim_{\e,\lmp{T}} x$ then $\mathrm{inl}(i),\mathrm{inr}(i),\mathrm{inl}(x)$ and $\mathrm{inr}(x)$ are all in the same $(\sim_{\s,\lmp{T}\oplus\lmp{T}_\e})$-equivalence class. 
\end{corollary}
\begin{proof}
	If $i \sim_{\e,\lmp{T}} x$, by Lemma~\ref{lem:event-lmp}, we have $i \sim_{\s,\lmp{T}_\e} x$. 
	Thanks to Lemma~\ref{lem:sum-preserves-bisimil}, we conclude $\mathrm{inr}(i) \sim_{\s,\lmp{T}\oplus\lmp{T}_\e} \mathrm{inr}(x)$.
	We now notice that the relation $\{(\mathrm{inl}(t),\mathrm{inr}(t))\mid t\in T\}$ in $\lmp{T}\oplus \lmp{T}_\e$ is a state bisimulation as it is the lifting of the graph of the zigzag morphism $\id:\lmp{T}\to \lmp{T}_\e$ (Lemma~\ref{lem:graph(f)-bisim} and Proposition~\ref{prop:ext-equiv-oplus}(\ref{item:ext-equiv-oplus-bisim})).
	From this and the fact that $\sim_{\sfs, \lmp{T}\oplus\lmp{T}_\e}$ is an equivalence we have to the result.
\end{proof}

\begin{theorem}
	$s\sim^\catco s'$ if and only if they are $\oplus_\mathsf{P}$-bisimilar.
\end{theorem}
\begin{proof}	
	If $s\sim^\catco s'$, via Theorem~\ref{th:catco_eq_event} we have that $s\sim_{\e,\lmp{S}\oplus\lmp{S}'} s'$.
	Corollary~\ref{coro:event-bisim-implies-state-in-sum-event} applied to $\lmp{T}=\lmp{S}\oplus\lmp{S}'$, says that they are state bisimilar in $(\lmp{S}\oplus \lmp{S}')\oplus(\lmp{S}\oplus\lmp{S}')_\e$ (strictly speaking, $\mathrm{inl}(\mathrm{inl}(s)) \sim_{\sfs,(\lmp{S}\oplus \lmp{S}')\oplus(\lmp{S}\oplus\lmp{S}')_\e} \mathrm{inl}(\mathrm{inr}(s'))$).
	Then $s$ and $s'$ are $\oplus_\mathsf{P}$-bisimilar.
	
	For the converse, if $s$ and $s'$ are
        $\oplus_\mathsf{P}$-bisimilar then $\inl(s)$ and $\inr(\inl(s'))$ are state bisimilar in $\lmp{S}\oplus(\lmp{S}'\oplus\lmp{W})$ for some $\lmp{W}$, and therefore also event bisimilar in that process.
	By Theorem~\ref{th:catco_eq_event}, this implies that there is
        a cospan $\lmp{S}\rightarrow \lmp{T} \leftarrow
        \lmp{S}'\oplus\lmp{W}$ which identifies $s$ and $\inl(s')$.
	Composing this cospan with the zigzag inclusion $\inl:\lmp{S}'\to \lmp{S}'\oplus\lmp{W}$ leads to $s\sim^\catco s'$. 
\end{proof}

\subsection{Comparison of notions for LMP and related work}
\label{sec:comparison-notions-lmp}

Table~\ref{tab:bisim-comparison} summarizes the relationships between
the different definitions of bisimilarity between two LMPs $\lmp{S}$
and $\lmp{S}'$ studied up to this point.
The entry with indices $(i,j)$ contains the symbol {\cmark} if the bisimilarity
corresponding to row $i$ is included in the bisimilarity of
column $j$; it contains the symbol \xmark{} when the inclusion fails; or \textbf{?} with
its obvious meaning.  Each cell also includes an indication for the
justification of each assertion, which may be a reference to a result,
a word suggesting the answer, or a reference to another cell $(i,j)$
in the same table ($(2,5)$ corresponds to the cell ${\sim}^\catsp \sbq
{\sim}^\catco$). Note that there are
some double-symbol entries in the table, in which references are
given for the known inclusions or their failures, either in the literature or as
a result proved here.

\begin{table}[h]
  \scriptsize
  \centering
  \resizebox{\textwidth}{!}{%
    \begin{tabular}{|c|c|c|c|c|c|c|}
      \hline
      $\sbq$ &
      \rule{0pt}{2.5ex}$\sim^\coalg$ &
      $\sim^\catsp$ &
      $\sim^\times$ &
      $({\sim}_\sfs)_\times$ &
      $\sim^\catco$ &
      $\sim^\oplus$ 
      
      \\ \hline

      $\sim^\coalg$ &

      \cmark &

      \begin{tabular}[c]{@{\!\!\!}c@{\!\!\!}}\cmark\\ by definition\end{tabular} &

      \begin{tabular}[c]{@{\!\!\!}c@{\!\!\!}}\cmark\\ \cite{bacci}\end{tabular} &

      \begin{tabular}[c]{@{\!\!\!}c@{\!\!\!}}\cmark\\ (1,2) \& (2,4)\end{tabular} &

      \begin{tabular}[c]{@{\!\!\!}c@{\!\!\!}}\cmark\\ (1,2) \& (2,5)\end{tabular} &

      \begin{tabular}[c]{@{\!\!\!}c@{\!\!\!}}\cmark\\ (1,2) \& (2,6)\end{tabular}

      \\ \hline

      $\sim^\catsp$ &

      \begin{tabular}[c]{@{\!\!\!}c@{\!\!\!}}\textbf{?} - \ref{question:span-implies-coalg}\\ \cmark \ - \ref{th:coalg-cat-coanalytic}, \cite{Doberkat:2005:SSR:1089905.1089907}\end{tabular} &

      \cmark &

      \begin{tabular}[c]{@{\!\!\!}c@{\!\!\!}}\cmark\\   \ref{prop:span-implies-ext}\end{tabular} &

      \begin{tabular}[c]{@{\!\!\!}c@{\!\!\!}}\cmark\\ (2,3) \& (3,4)\end{tabular} &

      \begin{tabular}[c]{@{\!\!\!}c@{\!\!\!}}\cmark\\ \cite{coco}: pushouts\end{tabular}  &

      \begin{tabular}[c]{@{\!\!\!}c@{\!\!\!}}\cmark\\ (2,3) \& (3,6)\end{tabular}

      \\ \hline

      $\sim^\times$ &

      \begin{tabular}[c]{@{\!\!\!}c@{\!\!\!}}\xmark\\        (3,2) \&        (1,2)\end{tabular} &
      
      \begin{tabular}[c]{@{\!\!\!}c@{\!\!\!}}\xmark\\         \ref{exm:ext-not-cat} \end{tabular} &

      \cmark &

      \begin{tabular}[c]{@{\!\!\!}c@{\!\!\!}}\cmark\\         \ref{cor:ext-bis-sbq-R-oplus}\end{tabular} &

      \begin{tabular}[c]{@{\!\!\!}c@{\!\!\!}}\cmark\\  (3,4) \& (4,5)\end{tabular} &

      \begin{tabular}[c]{@{\!\!\!}c@{\!\!\!}}\cmark\\  (3,4) \& (4,6)\end{tabular}

      \\ \hline

      $({\sim}_\sfs)_\times$ &

      \begin{tabular}[c]{@{\!\!\!}c@{\!\!\!}}\xmark\\           (4,2) \& (1,2)\end{tabular} &

      \begin{tabular}[c]{@{\!\!\!}c@{\!\!\!}}\xmark\\         (3,2) \& (3,4) \end{tabular} &

      \begin{tabular}[c]{@{\!\!\!}c@{\!\!\!}}\textbf{?} - \ref{question:equiv-int-sum-implies-ext}\\ \cmark \ - \ref{cor:R-oplus-ext}\end{tabular} &

      \cmark&

      \begin{tabular}[c]{@{\!\!\!}c@{\!\!\!}}\cmark\\         \ref{cor:state-impl-catco}\end{tabular} &

      \begin{tabular}[c]{@{\!\!\!}c@{\!\!\!}}\cmark\\      \ref{cor:R-oplus-sbq-sim-oplus}\end{tabular}  

      \\ \hline

        $\sim^\catco$ &

      \begin{tabular}[c]{@{\!\!\!}c@{\!\!\!}}\xmark \\ (5,4) \& (1,4) \end{tabular} &

        \begin{tabular}[c]{@{\!\!\!}c@{\!\!\!}}\xmark \ -     \cite{Pedro20111048}\\ \cmark \ - 
          \ref{lem:semi-pullbacks} \end{tabular} &

        \begin{tabular}[c]{@{\!\!\!}c@{\!\!\!}}\xmark\\ (5,4) \&  (3,4)\end{tabular} &

        \begin{tabular}[c]{@{\!\!\!}c@{\!\!\!}}\xmark\\  \ref{exm:catco-not-ext}\end{tabular}&

        \cmark &

        \begin{tabular}[c]{@{\!\!\!}c@{\!\!\!}}\textbf{?}\\  \ref{question:oplus-cospan-unknown}\end{tabular}

      \\ \hline

      $\sim^\oplus$ &

      \begin{tabular}[c]{@{\!\!\!}c@{\!\!\!}}\xmark\\ (6,5) \& (1,5) \end{tabular} &

      \begin{tabular}[c]{@{\!\!\!}c@{\!\!\!}}\xmark\\ (6,5) \& (2,5) \end{tabular} &

      \begin{tabular}[c]{@{\!\!\!}c@{\!\!\!}}\xmark\\ (6,5) \& (3,5) \end{tabular} &

      \begin{tabular}[c]{@{\!\!\!}c@{\!\!\!}}\xmark \\ (6,5) \& (4,5)  \end{tabular} &

      \begin{tabular}[c]{@{\!\!\!}c@{\!\!\!}}\xmark \\ \ref{exm:oplus-notsbq-cospan}\end{tabular} &

      \cmark

      \\ \hline
    \end{tabular}%
  }
  \caption{Comparison between bisimilarity notions.}
  \label{tab:bisim-comparison}
\end{table}

It is worth noting that some of those results that follow “by transitivity” in
our diagram have appeared in previous works, but not with the exact same
statement, or only with a hint of proof---our generous sample of
counterintuive examples ask for detailed proofs of all alleged results. For
instance, cell $(2,4)$ is implicit in Danos et al. \cite[p.515]{coco}, but in that paper, all
zigzag morphisms are considered to be surjective. In the same direction, $(2,3)$
is somewhat implicit in the proof of $(1,3)$ from Bacci \cite{bacci}, but otherwise
taken for granted.  In any case, our Proposition~\ref{prop:span-implies-ext}
invoked for that proof is also new, in that it provides a “local” argument
(which holds for each pair of states) and again does not require
surjectivity. The weaker Lemma~\ref{lem:graph(f)-bisim} also seems to be new (in
this form); there are some antecedents to it in
\cite[Prop.~3.6, Thm.~4.5]{2014arXiv1405.7141D} but these only refer to internal
notions and require symmetry of the bisimulation relation.

There are some relations between already known notions that do not appear
in the table since they are actually \textit{equalities}. For example,
Theorem~\ref{th:catco_eq_event} shows that $\sim^\catco$ is the
same as (internal) event bisimilarity in the sum; the same observations from the
previous paragraph apply to this equality. The other notion we proved to be equal to
$\sim^\catco$ is Panangaden's $\oplus_{\mathrm{P}}$-bisimilarity.

\section{Nondeterministic processes}
\label{sec:NLMP}

Nondeterministic labelled
Markov processes (NLMPs) were introduced in
\cite{DWTC09:qest,Wolovick} as an alternative to include internal
nondeterminism in LMP.  This is achieved by
modifying the codomain of the transition functions, now denoted by
$\Tfont{T}_a$, to return a \textit{measurable set} of subprobability
measures for each state, as
opposed to an LMP where there is only one measure for each label.  
Formally, let  $\Delta(\Sigma)$ be the smallest $\sigma$-algebra that makes all
the evaluation functions $\mathrm{ev}_E:\Delta(S) \to [0,1]$ measurable with
respect to $\Borel([0, 1])$. Then the values of $\Tfont{T}_a$ range over
$\Delta(\Sigma)$.

In the following, we will need some notation for certain
prominent sets of measures: If $E\in \Sigma$ and ${ \bowtie } \in
\{<, \leq, >,\geq \}$, $\Delta^{\bowtie q}(E)$ denotes the set
$\{ \mu\in  \Delta(S) \mid \mu(E)\bowtie q \}$.
The family $\mathcal{C}_0$, consisting of all these sets for ${\bowtie} = {>}$,
$q\in \Q\cap [0,1]$, and $E\in \Sigma$, generates the $\sigma$-algebra.
In general, if $\Gamma \sbq \Sigma$, $\Delta(\Gamma)\defi \sigma( \{\Delta^{>q}(E)\mid E\in \Gamma \})$.

The
definition of $\Tfont{T}_a$ also requires such transition functions to be measurable,
but this first demands a $\sigma$-algebra over the set
$\Delta(\Sigma)$.

\begin{definition}\label{def:hit-algebra}
  Given a set $X$ and a family $\Gamma\sbq \Power(X)$, the \emph{hit
  $\sigma$-algebra} $H(\Gamma)$ is the smallest
  $\sigma$-algebra on $\Gamma$ that contains all sets $H_D\defi \{G\in
    \Gamma\mid G\cap D \neq\emptyset\}$ with $D \in \Gamma$.
\end{definition}

The candidate $\sigma$-algebra on $\Delta(\Sigma)$ is then
$H(\Delta(\Sigma))$.  This choice is motivated by the
need to have measurable validity sets for the formulas
(\ref{eq:logic}) of the modal logic $\Logic$.
In the case of a formula with modality $\posbow{a}$, the 
standard semantics yields
\begin{align*}
  \sem{\posbow{a}\varphi}
  & =\{s\in S\mid \exists \mu\in
  \Tfont{T}_a(s)\; \mu(\sem{\varphi})\bowtie q\}\\
  & =\{s\in S\mid
  \Tfont{T}_a(s)\cap \Delta^{\bowtie q}(\sem{\varphi})\neq
  \emptyset\}\\
  & =\Tfont{T}_a^{-1}\bigl[\{D\in \Delta(\Sigma)\mid D\cap
    \Delta^{\bowtie q}(\sem{\varphi})\neq \emptyset\}\bigr]
    =\Tfont{T}_a^{-1}[H_{\Delta^{\bowtie q}(\sem{\varphi})}].
\end{align*}
Therefore, if $\sem{\varphi}$ is measurable, the measurability of $\sem{\posbow{a}\varphi}$ is
guaranteed since then  $H_{\Delta^{\bowtie q}(\sem{\varphi})}\in
H(\Delta(\Sigma))$.

\begin{definition}\label{def:nlmp}
  A \emph{nondeterministic labelled Markov process}, or NLMP, is a structure $\lmp{S} 
  =(S,\Sigma,\{\Tfont{T}_a\mid a\in L\})$ where $\Sigma$ is 
  a $\sigma$-algebra over the state set $S$, and 
  for each label $a\in L$, $\Tfont{T}_a:(S,\Sigma)\to 
  (\Delta(\Sigma),H(\Delta(\Sigma)))$ is measurable. 
  We say that $\lmp{S}$ is \emph{image-finite (countable)} if all sets $\Tfont{T}_a(s)$ are finite (countable). 
  NLMPs that are not image-countable are called \emph{image-uncountable}.
\end{definition}

Originally, the definition of NLMP in \cite{DWTC09:qest} required that
the measures be \textit{probabilities} except that the null kernel
$\tau_a(s,\cdot)=0$ was allowed in order to indicate that the action $a$
was rejected in state $s$ with certainty. In this work, we will
continue working with subprobabilities to allow a smooth
generalization of LMP.

\begin{example}[LMP as NLMP]
  An LMP $\lmp{S}$ can be seen as an NLMP without internal
  nondeterminism, that is, a process where $\Tfont{T}_a(s)$ is a
  singleton for each $a\in L$ and $s\in S$. The only requirement is
  that the sets $\{ \mu \}$ are measurable in $\Delta(S)$, so
  $\Tfont{T}_a(s)=\{\tau_a(s,\cdot)\}\in \Delta(\Sigma)$. A sufficient
  condition is that $\Sigma$ is generated by a countable algebra. A
  modification of the proof of \cite[Lem.~3.1]{DWTC09:qest} using
  the Monotone Class Theorem provides a proof.
\end{example}

In \cite{DWTC09:qest, D'Argenio:2012:BNL:2139682.2139685}, the
following three notions of bisimulation are defined. The names are
revised according to the newer \cite{ROCKS}. Recall the relation
lifting defined by~(\ref{eq:int-lift}).

\begin{definition}\label{def:nlmp-bisimulations}
  \begin{enumerate}
    \item \label{def:nlmp_sb}
    A relation $R \sbq S\times S$ is a \emph{state bisimulation} on an NLMP $(S,\Sigma,\{\Tfont{T}_a\mid a\in L\})$ if it is symmetric and for all $a \in L$, $s\mathrel{R}t$ implies that for every $\mu \in \Tfont{T}_a(s)$ there exists $\mu'\in \Tfont{T}_a(t)$ such that $\mu \mathrel{\bar{R}}\mu'$.
    
    \item \label{def:nlmp_hb}
    A relation $R\sbq S\times S$ is an \emph{hit bisimulation} if it is symmetric and for all $a\in L$, $s\mathrel{R}t$ implies $\forall \xi\in \Delta(\Sigma(R))$, $\Tfont{T}_a(s)\cap \xi\neq \emptyset \iff \Tfont{T}_a(t)\cap\xi\neq \emptyset$.
    
    \item \label{def:nlmp_eb} 
    An \emph{event bisimulation} is a sub-$\sigma$-algebra $\Lambda$ of $\Sigma$ such that the map $\Tfont{T}_a: (S,\Lambda)\to(\Delta(\Sigma),H(\Delta(\Lambda)))$ is measurable for each $a\in L$. We also say that a relation $R$ is an event bisimulation if there is an event bisimulation $\Lambda$ such that $R=\Rel(\Lambda)$.
  \end{enumerate}
  We say that $s,t\in S$ are state (hit, event) \emph{bisimilar}, denoted as $s\mathrel{\sim_\sfs}t$ ($s\mathrel{\sim_\h}t$, $s\mathrel{\sim_\e}t$), if there is a state (hit, event) bisimulation $R$ such that $s\mathrel{R}t$.
\end{definition}

Note that all these notions of bisimulation/bisimilarity are defined
with respect to a single NLMP, or, using the terminology for LMP, they
are \textit{internal}. Event bisimulation is a direct generalization
of the same concept for LMP and best reflects the measurable structure
of the base space. State bisimulation is the most faithful
generalization of both probabilistic bisimulation by Larsen and Skou
and the standard definition of bisimulation for nondeterministic
processes (e.g., LTS). On the other hand, hit bisimilarity is an
intermediate concept between the previous ones. It is generally finer
than event bisimilarity but respects the measurable structure more
than state bisimilarity. To ensure that $R$ is a hit bisimulation,
it is required that the transition sets intersect the same
$\bar{R}$-closed sets of measures. This is a weaker requirement than
that of state bisimulation. For further discussion on these notions,
we refer to \cite{D'Argenio:2012:BNL:2139682.2139685}.

\begin{theorem}[{\cite[Thm.~5.6, Thm.~5.10]{D'Argenio:2012:BNL:2139682.2139685}}]
  \begin{enumerate}
    \item The inclusions ${\sim_\sfs} \subseteq {\sim_\h} \subseteq
      {\sim_\e}$ hold, and both can be strict.
    \item If $\lmp{S}$ is an image-finite NLMP over an analytic space, all bisimilarities coincide.
  \end{enumerate}
\end{theorem}

\subsection{External bisimulations for NLMP}
\label{sec:NLMP-external}

In this section, we  provide external versions of all
bisimulations in Definition~\ref{def:nlmp-bisimulations}.
Given a relation $R \subseteq S \times
S'$, we reuse the Definition~\ref{def:closed-pair} of $R$-closed pair
and, for convenience, use the notation $\Sigma^\times(R)$ for the
family of $R$-closed measurable pairs $(Q, Q')$. We extend the
definition of lift (\ref{eq:int-lift}) to relations $R$ on $\Delta(S) \times \Delta(S')$ as follows:
\begin{equation}\label{eq:ext-lift}
  \mu \mathrel{\bar{R}} \mu' \iff \forall (Q, Q') \in
  \Sigma^{\times}(R) \; \mu(Q) = \mu'(Q').
\end{equation}
We also use the analogous notation $\Delta(\Sigma)^\times (\bar{R})$
for the family of measurable pairs $(\Theta, \Theta')$ that are
$\bar{R}$-closed.

\begin{remark}\label{note:closed-pairs-reflexive}
  By
  Lemma~\ref{lem:closed-vs-closed-pair}(\ref{item:closed-pair-reflexive-case}),
  if $\lmp{S} = \lmp{S'}$ and $R$ is reflexive, then $(Q, Q') \in
  \Sigma^\times(R) \iff Q = Q' \in \Sigma(R)$. Therefore, the liftings
  \eqref{eq:int-lift} and \eqref{eq:ext-lift} essentially coincide
  (modulo the correspondence $x\longleftrightarrow (x,x)$).
\end{remark}

The following lemma states elementary properties of lifting relations
and follows immediately from Lemma~\ref{lem:pair-equiv} and the
definitions.

\begin{lemma}\label{lem:R-closed-pair-prop}
  \begin{enumerate}
  \item\label{item:closed-pair-closure-prop} The family of
    $\bar{R}$-closed pairs is closed under coordinatewise
    complementation, and coordinatewise arbitrary unions and
    intersections.
  \item\label{item:Sigma-times-antimonotone} If $R \subseteq R'$,
    then $\bar{R} \subseteq \bar{R'}$.
  \item\label{item:closed-pair-implies-bar-closed-pair} If $(Q, Q')
    \in \Sigma^\times(R)$, then $(\Delta^{<q}(Q), \Delta^{<q}(Q'))
    \in \Delta(\Sigma)^\times (\bar{R})$.
  \end{enumerate}
\end{lemma}

We also need the following construction for families of pairs:

\begin{definition}\label{def:ext-cl-operator}
  Let $\mathcal{D} \subseteq \Power(S) \times \Power(S')$ be a family
  of pairs of sets. We define $\sigma^\times(\mathcal{D})$ as the
  smallest family $\Fam \subseteq \Power(S) \times \Power(S')$ such
  that
  \begin{enumerate}
  \item $\mathcal{D} \subseteq \Fam$.
  \item $(E,F) \in \Fam \implies (E^c,F^c) \in \Fam$.
  \item $\forall n \in \omega \; (E_n, F_n) \in \Fam \implies
    (\bigcup_{n \in \omega} E_n, \bigcup_{n \in \omega} F_n) \in
    \Fam$.
  \end{enumerate}
  We call any non-empty subfamily $\Fam$ of $\Power(S) \times
  \Power(S')$ that satisfies conditions 2 and 3 a
  \emph{bi-$\sigma$-algebra}.
\end{definition}

The projection of a
bi-$\sigma$-algebra onto either of its coordinates is a
$\sigma$-algebra.

\begin{definition}\label{def:Delta-pairs}
  If $(S,\Sigma)$ and $(S',\Sigma')$ are measurable spaces and
  $\mathcal{D} \subseteq \Sigma \times \Sigma'$, we define
  $\Delta^\times(\mathcal{D}) \subseteq \Power(\Delta(S)) \times
  \Power(\Delta(S'))$ as
  \[\Delta^\times(\mathcal{D}) \defi \sigma^\times(\{(\Delta^{<q}(Q), \Delta^{<q}(Q')) \mid q \in \Q, (Q,Q') \in \mathcal{D}\}).\]
\end{definition}

\begin{remark}\label{note:Delta-times-prop}
  Note that $\Delta^\times(\mathcal{D}) \subseteq \Delta(\Sigma)
  \times \Delta(\Sigma')$, since $\Fam = \Delta(\Sigma) \times
  \Delta(\Sigma')$ satisfies all the conditions in
  Definition~\ref{def:ext-cl-operator}.
\end{remark}

\begin{lemma}\label{lem:Delta-reflexive-pair}
  If $(S,\Sigma) = (S',\Sigma')$ in Definition~\ref{def:Delta-pairs},
  and $\mathcal{D} = \{(D,D) \mid D \in \A\}$ for some family $\A
    \subseteq \Sigma$, then $\Delta^\times(\mathcal{D}) =
    \{(\Theta,\Theta) \mid \Theta \in \Delta(\A)\}$.
\end{lemma}
\begin{proof}
  $(\subseteq)$ Let $\Fam = \{(\Theta,\Theta) \in \Delta(\Sigma)
    \times \Delta(\Sigma) \mid \Theta \in \Delta(\A)\}$. We will show
  that $\Delta^\times(\mathcal{D}) \subseteq \Fam$. If $D \in \A$,
  then $\Theta \defi \Delta^{<q}(D) \in \Delta(\A)$. Therefore,
  \[
    \{(\Delta^{<q}(D), \Delta^{<q}(D)) \mid q \in \Q, \; (D,D) \in
    \mathcal{D}\} \subseteq \Fam.
  \]
  If $(\Theta,\Theta),
  (\Theta_n,\Theta_n) \in \Fam$, then $\Theta, \Theta_n \in
  \Delta(\A)$, and since this is a $\sigma$-algebra, we have
  $\Theta^c, \bigcup_{n \in \omega} \Theta_n \in \Delta(\A)$. Hence,
  $(\Theta^c, \Theta^c), (\bigcup_{n \in \omega} \Theta_n, \bigcup_{n
    \in \omega} \Theta_n) \in \Fam$. We conclude that
  $\Delta^\times(\mathcal{D}) \subseteq \Fam$, as desired.
  
  $(\supseteq)$ Let $\mathcal{H} = \{\Theta \in \Delta(\Sigma) \mid
    (\Theta,\Theta) \in \Delta^\times(\mathcal{D})\}$. Since
  $\Delta^\times(\mathcal{D})$ is closed under complement and
  countable unions in each coordinate, $\mathcal{H}$ is a
  $\sigma$-algebra. Moreover, if $D \in \A$, then $(\Delta^{<q}(D),
  \Delta^{<q}(D)) \in \Delta^\times(\mathcal{D})$ for any $q \in \Q$;
  hence, $\Delta^{<q}(D) \in \mathcal{H}$. It follows that $\Delta(\A)
  = \sigma(\{\Delta^{<q}(D) \mid q \in \Q, \; D \in \A\}) \subseteq
    \mathcal{H}$.
\end{proof}

\begin{corollary}\label{cor:measures-closed-pairs-single-NLMP}
  Let $R \subseteq S \times S$ be reflexive. Then $(\Theta,\Theta')
  \in \Delta^\times(\Sigma^\times(R))$ if and only if $\Theta =
  \Theta' \in \Delta(\Sigma(R))$.
\end{corollary}
\begin{proof}
  By Remark~\ref{note:closed-pairs-reflexive}, $\Sigma^\times(R)
  = \{(Q,Q) \mid Q \in \Sigma(R)\}$. We use
  Lemma~\ref{lem:Delta-reflexive-pair} to obtain
  $\Delta^\times(\Sigma^\times(R)) = \{(\Theta,\Theta) \mid \Theta \in
    \Delta(\Sigma(R))\}$.
\end{proof}

Another consequence of Lemma~\ref{lem:Delta-reflexive-pair} is that if
$\pi$ is any of the projections, then $\pi[\Delta^\times(\{(D,D) \mid
    D \in \A\})] = \Delta(\A) = \Delta(\pi[\{(D,D) \mid D \in
    \A\}])$. This fact holds in general, as shown in the following
result that we prove for the projection $\pi$ on the first coordinate
(it holds \textit{mutatis mutandis} for the projection on the second
coordinate).

\begin{lemma}\label{lem:projection-Delta-times}
  If $\mathcal{D} \subseteq \Sigma \times \Sigma'$, then
  $\pi[\Delta^\times(\mathcal{D})] = \Delta(\pi[\mathcal{D}])$.
\end{lemma}
\begin{proof}
  $(\subseteq)$ Note that $\Delta(\pi[\mathcal{D}]) \times
  \Delta(\Sigma')$ is a bi-$\sigma$-algebra that contains the
  generators
  of $\Delta^\times(\mathcal{D})$. Therefore,
  $\Delta^\times(\mathcal{D}) \subseteq \Delta(\pi[\mathcal{D}]) \times
  \Delta(\Sigma')$, and hence $\pi[\Delta^\times(\mathcal{D})]
  \subseteq \Delta(\pi[\mathcal{D}])$.
  
  $(\supseteq)$ We know that $\pi[\Delta^\times(\mathcal{D})]$ is a
  $\sigma$-algebra. If $Q \in \pi[\mathcal{D}]$, then there exists $Q'
  \in \Sigma'$ such that $(Q,Q') \in \mathcal{D}$. Consequently,
  $(\Delta^{<q}(Q),\Delta^{<q}(Q')) \in \Delta^\times(\mathcal{D})$,
  and therefore $\Delta^{<q}(Q) \in
  \pi[\Delta^\times(\mathcal{D})]$. We conclude that
  $\Delta(\pi[\mathcal{D}]) \subseteq \pi[\Delta^\times(\mathcal{D})]$.
\end{proof}

The following Lemma tells us that if $\mathcal{C}$ consists of
measurable $R$-closed pairs, then $\Delta^\times(\mathcal{C})$
consists of $\Delta(\Sigma)$-measurable $\bar{R}$-closed pairs.

\begin{lemma}\label{lem:closed-pair-implies-bar-closed-pair}
  If $\mathcal{C}\subseteq \Sigma^{\times}(R)$, then
  $\Delta^\times(\mathcal{C})\subseteq \Delta(\Sigma)^\times(\bar{R})$.
\end{lemma}
\begin{proof}
  If $(Q,Q')$ is a measurable $R$-closed pair, by
  Lemma~\ref{lem:R-closed-pair-prop}(\ref{item:closed-pair-implies-bar-closed-pair}),
  \[
    (\Delta^{<q}(Q),\Delta^{<q}(Q')) \in
    \Delta(\Sigma)^\times(\bar{R}).
  \]
  Additionally, by
  Lemma~\ref{lem:R-closed-pair-prop}(\ref{item:closed-pair-closure-prop}),
  $\Delta(\Sigma)^\times(\bar{R})$ is a bi-$\sigma$-algebra.  This shows
  that $\mathcal{F}\defi \Delta(\Sigma)^\times(\bar{R})$ satisfies the
  properties in Definition~\ref{def:ext-cl-operator} for the generating
  family of $\Delta^\times(\Sigma^\times(R))$.  Consequently,
  $\Delta^\times(\Sigma^\times(R)) \subseteq
  \Delta(\Sigma)^\times(\bar{R})$.  If now $\mathcal{C}\subseteq
  \Sigma^\times(R)$, then the inclusions $\Delta^\times(\mathcal{C})
  \subseteq \Delta^\times(\Sigma^\times(R)) \subseteq
  \Delta(\Sigma)^\times(\bar{R})$ hold.
\end{proof}

Using these concepts, we can propose the following definitions for
bisimulations between two different NLMP.

\begin{definition}\label{def:nlmp-ext-bisimulations}
  \begin{enumerate}
  \item \label{def:nlmp-ext-sb} A relation $R \subseteq S\times S'$ is
    an \emph{external state bisimulation} if for every $a \in L$,
    $s\mathrel{R}s'$ implies that for every $\mu \in \Tfont{T}_a(s)$
    there exists $\mu'\in \Tfont{T}'_a(s')$ such that $\mu
    \mathrel{\bar{R}}\mu'$ (“zig”) and for every $\mu' \in \Tfont{T'}_a(s')$
    there exists $\mu\in \Tfont{T}_a(s)$ such that $\mu
    \mathrel{\bar{R}}\mu'$ (“zag”).
    
  \item \label{def:nlmp-ext-hb} A relation $R\subseteq S\times S'$ is
    an \emph{external hit bisimulation} if for every $a\in L$,
    $s\mathrel{R}s'$ implies that for every pair $(\Theta,\Theta')\in
    \Delta^\times(\Sigma^{\times}(R))$ we have $\Tfont{T}_a(s)\cap
    \Theta\neq \emptyset \iff \Tfont{T}_a(t)\cap \Theta'\neq
    \emptyset$.
  \end{enumerate}
  We will say that $s \in S, s'\in S'$ are \emph{ext-state bisimilar}
  (\emph{ext-hit bisimilar}), denoted as $s\sim_\sfs^\times s'$
  ($s\mathrel{\sim_\h^\times}s'$), if there exists an external state
  (hit) bisimulation $R$ such that $s\mathrel{R}s'$.
\end{definition}

\begin{lemma}
  If $\lmp{S}=\lmp{S}'$, the external definitions of state and hit
  bisimulations coincide with the internal ones in
  the reflexive case.
\end{lemma}
\begin{proof}
  In the case of state bisimulations,
  Remark~\ref{note:closed-pairs-reflexive} proves that if $R$ is
  reflexive, then it is an external state bisimulation if and only if
  it is an internal state bisimulation.
  
  Suppose that $R\subseteq S\times S$ and $s\mathrel{R} s'$. If $R$ is
  an external hit bisimulation and $\Theta \in \Delta(\Sigma(R))$,
  then $(\Theta, \Theta)\in \Delta^\times(\Sigma^\times(R))$ by
  Corollary~\ref{cor:measures-closed-pairs-single-NLMP}. Therefore,
  $\Tfont{T}_a(s)\cap \Theta\neq \emptyset \iff \Tfont{T}_a(t)\cap
  \Theta\neq \emptyset$. Conversely, if $R$ is an internal hit
  bisimulation and $(\Theta, \Theta')\in
  \Delta^\times(\Sigma^\times(R))$ then $\Theta=\Theta'\in
  \Delta(\Sigma(R))$ by the same Corollary, and it holds that
  $\Tfont{T}_a(s)\cap \Theta\neq \emptyset \iff \Tfont{T}_a(t)\cap
  \Theta'\neq \emptyset$.
\end{proof}

If an external state bisimulation $R$ in the previous lemma is not
reflexive, it is straightforward to verify that $R\cup R^{-1}$ is an
internal state bisimulation.

\begin{prop}\label{prop:union-ext-bisimulation}
  The union of an arbitrary family of external hit (state)
  bisimulations is an external hit (state) bisimulation. Therefore,
  $\sim_\h^\times$ ($\sim_\sfs^\times$) is an external hit (state)
  bisimulation.
\end{prop}

\begin{proof}
  Let $\{R_i\mid i\in I\}$ be a family of relations, and let $R =
    \bigcup\{R_i\mid i\in I\}$. If $s\mathrel{R} s'$, there exists
    $i\in I$ such that $s\mathrel{R_i}s'$. Since $R_i\subseteq R$, by
    Lemmas~\ref{lem:pair-equiv}(\ref{item:closed-pairs-antimonotonicity})
    and
    \ref{lem:R-closed-pair-prop}(\ref{item:Sigma-times-antimonotone}),
    we have:
    \begin{enumerate}
    \item $\Sigma^\times(R)\subseteq \Sigma^\times(R_i)$, and thus
      $\Delta^\times(\Sigma^\times(R)) \subseteq
      \Delta^\times(\Sigma^\times(R_i))$.
    \item $\bar{R_i}\subseteq \bar{R}$.
    \end{enumerate}
    If $R_i$ is a hit bisimulation and $(\Theta,\Theta')\in
    \Delta(\Sigma^\times(R))$, then by the first item above,
    $(\Theta,\Theta')\in \Delta(\Sigma^\times(R_i))$, and we have
    $s\in\Tfont{T}_a^{-1}[H_\Theta]\iff s'\in
    \Tfont{T}'_a\!^{-1}[H_{\Theta'}]$.

    If $R_i$ is a state bisimulation and $\mu \in \Tfont{T}_a(s)$,
    there exists $\mu'\in \Tfont{T}'_a(s')$ such that $\mu
    \mathrel{\bar{R_i}}\mu'$. By the second item above, $\mu
    \mathrel{\bar{R}}\mu'$. The zag condition is analogous.
\end{proof}

\begin{lemma}\label{lem:ext-state-implies-hit-NLMP}
  If $R$ is an external state bisimulation, then it is also an
  external hit bisimulation.
\end{lemma}

\begin{proof}
  Let $(s,s')\in R$ and $(\Theta,\Theta')\in
  \Delta^\times(\Sigma^\times(R))$. By
  Lemma~\ref{lem:closed-pair-implies-bar-closed-pair},
  $(\Theta,\Theta')$ is a $\bar{R}$-closed measurable pair. If there
  exists $\mu \in\Tfont{T}_a(s)$ such that $\mu\in \Theta$, then since
  $R$ is a state bisimulation, there exists $\mu' \in
  \Tfont{T}'_a(s')$ such that $\mu\mathrel{\bar{R}}\mu'$. Therefore,
  $\mu' \in \Tfont{T}'_a(s')\cap \Theta'$. The converse implication
  follows by using the zag condition.
\end{proof}

The following lemma indicates that the converse to
Lemma~\ref{lem:ext-state-implies-hit-NLMP} holds for image-countable
NLMPs, just like in the internal case
\cite[Thm.4.5]{D'Argenio:2012:BNL:2139682.2139685}.

\begin{lemma}\label{lem:ext-impac-is-state-NLMP}
  If $\lmp{S}$ and $\lmp{S'}$ are image-countable NLMPs, then
  every external hit bisimulation is also an external state
  bisimulation.
\end{lemma}

\begin{proof}
  Suppose that $R\subseteq S\times S'$ is an external hit
  bisimulation, and let $(s,t)\in R$ such that
  $s\nsim_\sfs^\times t$. Without loss of generality,
  we can assume that there exists $\mu \in \Tfont{T}_a(s)$ such that
  for every $\mu'_n$ in an enumeration $\{ \mu'_m \}_{m\in \omega}$ of
  $\Tfont{T}'_a(t)$, there exist $q_n\in \Q$, ${\bowtie_n}\in
  \{ <,> \}$, and $(Q_n,Q'_n)$  $R$-closed measurable pairs
  such that $\mu(Q_n)\bowtie_n q_n \bowtie_n \mu'_n(Q'_n)$. Let
  $\Theta\defi \bigcap\{\Delta^{\bowtie_n q_n}(Q_n)\mid n\in
  \omega\}$ and $\Theta'\defi \bigcap\{\Delta^{\bowtie_n
    q_n}(Q'_n)\mid n\in \omega\}$. Then $(\Theta,\Theta')\in
  \Delta^\times(\Sigma^\times(R))$. By construction, $\mu \in
  \Theta$, and therefore $s\in
  \Tfont{T}_a^{-1}[H_\Theta]$. However,
  $t\notin\Tfont{T}'_a\!^{-1}[H_{\Theta'}]$ since for every $n\in
  \omega$, $q_n\bowtie_n \mu'_n(Q'_n)$. This contradicts the
  assumption that $R$ is a hit bisimulation.
\end{proof}

\begin{corollary}\label{cor:not-bisim-separation-NLMP}
  If $\lmp{S}$ and $\lmp{S'}$ are image-countable NLMPs and
  $s\nsim_\sfs^\times t$, then there exists a
  measurable $\sim_\sfs^\times$-closed pair $(C,C')$ such that $s\in C
  \iff t\notin C'$.
\end{corollary}

\begin{proof}
  We construct $\Theta,\Theta'$ as in the proof of
  Lemma~\ref{lem:ext-impac-is-state-NLMP} for $R={\sim_\sfs^\times}$,
  and we define $(C,C')\defi
  (\Tfont{T}_a^{-1}[H_\Theta],\Tfont{T}'_a\!^{-1}[H_{\Theta'}])$. It
  is clear that this pair is measurable. To verify that it is also
  $\sim_\sfs^\times$-closed, we will use the second equivalence in
  Lemma~\ref{lem:pair-equiv}(\ref{item:in-equiv}). Note that, by
  Proposition~\ref{prop:union-ext-bisimulation} and
  Lemma~\ref{lem:ext-state-implies-hit-NLMP}, $\sim_\sfs^\times$ is also
  a hit bisimulation. Therefore, for arbitrary $w$ and $z$, if
  $w\sim_\sfs^\times z$, then $w\in \Tfont{T}_a^{-1}[H_\Theta]\iff z\in
  \Tfont{T}'_a\!^{-1}[H_{\Theta'}]$.
\end{proof}

Finally, we present our version of external event bisimulations.

The motivation for the corresponding (internal) definition in LMP is to find the
smallest $\sigma$-algebra $\Lambda$ that preserves the LMP
structure. However, in a single NLMP, this property is lost because
the images of transition functions $\Tfont{T}_a(s)$ might still be
elements of $\Delta(\Sigma)$ and not $\Delta(\Lambda)$. Nevertheless, stability
still makes sense in the nondeterministic context.

\begin{definition}
  We say that $\mathcal{C} \subseteq \Sigma \times \Sigma'$ is
  \emph{$\times$-stable} if it satisfies
  \[\forall a \in L \; \forall (\Theta, \Theta') \in
    \Delta^\times(\mathcal{C}) \;
    (\Tfont{T}_a^{-1}[H_\Theta],\Tfont{T}_a'\!^{-1}[H_{\Theta'}]) \in \mathcal{C}.\]
\end{definition}
This definition extends its internal counterpart in the
following sense:
\begin{lemma}\label{lem:NLMP-stable-ext-properties}
  \begin{enumerate}
  \item $\mathcal{C} \subseteq \Sigma$ is stable if and only if
    $\tilde{\mathcal{C}}\defi\{(C,C)\mid C\in \mathcal{C}\}\subseteq
    \Sigma\times \Sigma$ is $\times$-stable.
  \item If $\mathcal{C} \subseteq \Sigma\times \Sigma'$ is a
    $\times$-stable bi-$\sigma$-algebra, then $\pi[\mathcal{C}]$ and
    $\pi'[\mathcal{C}]$ are stable $\sigma$-algebras.
  \end{enumerate}
\end{lemma}
\begin{proof}
  \begin{enumerate}
  \item If $(\Theta,\Theta') \in \Delta^\times(\tilde{\mathcal{C}})$, by
    Lemma~\ref{lem:Delta-reflexive-pair}, $\Theta = \Theta' \in
    \Delta(\mathcal{C})$. Since $\mathcal{C}$ is stable, we have
    $\Tfont{T}_a^{-1}[H_\Theta] \in \mathcal{C}$. Consequently,
    $(\Tfont{T}_a^{-1}[H_\Theta],\Tfont{T}_a'\!^{-1}[H_{\Theta'}]) \in
    \tilde{\mathcal{C}}$.

  \item We know that $\pi[\mathcal{C}]$ and $\pi'[\mathcal{C}]$ are
    $\sigma$-algebras. Let us prove that $\pi[\mathcal{C}]$ is stable. Let
    $\Theta \in \Delta(\pi[\mathcal{C}])$. By
    Lemma~\ref{lem:projection-Delta-times}, $\Theta \in
    \pi[\Delta^\times(\mathcal{C})]$. Therefore, there exists $\Theta'
    \in \Sigma'$ such that $(\Theta,\Theta') \in
    \Delta^\times(\mathcal{C})$, and since $\mathcal{C}$ is
    $\times$-stable, we have
    $(\Tfont{T}_a^{-1}[H_\Theta],\Tfont{T}'_a\!^{-1}[H_{\Theta'}]) \in
    \mathcal{C}$. Hence, $\Tfont{T}_a^{-1}[H_\Theta] \in
    \pi[\mathcal{C}]$. \qedhere
  \end{enumerate}
\end{proof}

\begin{prop}\label{prop:ext-hit-iff-ext-stable}
  Let $R \subseteq S\times S'$. Then, $R$ is an
  external hit bisimulation if and only if $\Sigma^\times (R)$ is
  $\times$-stable.
\end{prop}
\begin{proof}
  We have the equivalence
  \begin{multline*}
    \forall (s,s')\in R \; \bigl(\Tfont{T}_a(s)\cap\Theta \neq \emptyset
    \Leftrightarrow \Tfont{T}'_a(s')\cap\Theta' \neq \emptyset\bigr) \iff
    \\
    \iff (\Tfont{T}_a^{-1}[H_\Theta],(\Tfont{T}_a')^{-1}[H_{\Theta'}])\in \Sigma^\times(R).
  \end{multline*}
  Therefore, if $R$ is a hit bisimulation and $(\Theta,\Theta')\in
  \Delta^\times(\Sigma^\times(R))$, the forward direction shows that $\Sigma^\times(R)$ is stable. Conversely, if
  $\Sigma^\times(R)$ is stable, the backward direction proves that $R$
  is a hit bisimulation.
\end{proof}

\begin{definition}
  If $\D \subseteq \Power(S)\times\Power(S')$, we define a relation
  $\Rel^\times(\D) \subseteq S\times S'$ as
  \[ s\mathrel{\Rel^\times(\D)} s' \iff \forall (Q,Q') \in \D \; (s\in Q \Leftrightarrow s'\in Q').\]
\end{definition}

\begin{definition}\label{def:nlmp-ext-eb}
  A relation $R \subseteq S\times S'$ is an \emph{external event
  bisimulation} if there exists a  $\times$-stable bi-$\sigma$-algebra $\D
  \subseteq \Sigma\times \Sigma'$ such that $R = \Rel^\times(\D)$.
\end{definition}

We finish this section by relating external hit and event
bisimulations in analogous fashion to LMP \cite[Prop.~4.12]{coco} and
the internal NLMP notions
\cite[Lem.~4.2]{D'Argenio:2012:BNL:2139682.2139685}.

\begin{lemma}\label{lem:Rel-times-prop}
  \begin{enumerate}
  \item $R \subseteq \Rel^\times(\Sigma^\times(R))$.
  \item \label{item:Sigma-times-fixpoint} $\Sigma^\times(R) =
    \Sigma^\times(\Rel^\times(\Sigma^\times(R)))$.
  \end{enumerate}
\end{lemma}
\begin{proof}
  The first item follows directly from the definitions. For the second
  item, let us assume that $(Q,Q')\in \Sigma^\times(R)$ and
  $z\mathrel{\Rel^\times(\Sigma^\times(R))} w$. By definition of
  $\Rel^\times$, we have $z\in Q \iff w\in Q'$. Therefore, $(Q,Q')$ is an
  $\Rel^\times(\Sigma^\times(R))$-closed pair, and we have the
  inclusion $\Sigma^\times(R)\subseteq
  \Sigma^\times(\Rel^\times(\Sigma^\times(R)))$. For the reverse
  inclusion, we use the first item and the antimonotonicity of
  $\Sigma^\times$ given by
  Lemma~\ref{lem:pair-equiv}(\ref{item:closed-pairs-antimonotonicity}).
\end{proof}

\begin{corollary}\label{cor:nlmp-ext-imp-Rel-Sigma-also}
  If $R \subseteq S\times S'$ is an external hit bisimulation, then
  $\Rel^\times(\Sigma^\times(R))$ is a hit bisimulation and an
  external event bisimulation.
\end{corollary}
\begin{proof}
  To verify that it is a hit bisimulation, we take $(s,s')\in
  \Rel^\times(\Sigma^\times(R))$ and a pair
  $(\Theta,\Theta')\in
  \Delta(\Sigma^\times(\Rel^\times(\Sigma^\times(R))))$. By
  Lemma~\ref{lem:Rel-times-prop}(\ref{item:Sigma-times-fixpoint}), we
  have $(\Theta,\Theta')\in \Delta(\Sigma^\times(R))$ and, since $R$ is
  a hit bisimulation, $s\in \Tfont{T}_a^{-1}[H_\Theta]\iff s'\in
  \Tfont{T}'_a\!^{-1}[H_{\Theta'}]$. Additionally, by
  Proposition~\ref{prop:ext-hit-iff-ext-stable},
  $\Rel^\times(\Sigma^\times(R))$ is also an external event
  bisimulation.
\end{proof}

\subsection{Game characterization}
\label{sec:NLMP-games}

A typical problem when working with bisimulations between processes is
the question of characterizing bisimilarity in terms of games. For
basic concepts of the type of games we are interested in, we refer
to \cite[Sect~20.A]{Kechris}.
In \cite[p.23]{ic19}, a game is presented to characterize state
bisimilarity between two LMPs. Here, we
adapt such a game to characterize external state bisimilarity between
two %
NLMPs.

Fix two such processes
$\lmp{S}^0=(S^0,\Sigma^0,\{\Tfont{T}_a^0\mid a\in L \})$
and
$\lmp{S}^1=(S^1,\Sigma^1,\{\Tfont{T}_a^1\mid a\in L \})$.
The game begins with a pair of states $x_0^0\in S^0$,
$x_0^1\in S^1$, and unfolds according to the following rules:

\begin{enumerate}
\item
  Spoiler plays a label $a\in L$, an index $i\in \{ 0,1 \}$, a measure
  $\mu\in \Tfont{T}_a^i(x_j^i)$, fixes an enumeration
  $\Tfont{T}_a^{1-i}(x_j^{1-i})=\{ \mu'_k \mid k\in D \}$, and for
  each $k\in D$ a measurable pair
  $(C_k,C'_k)\in\Sigma^i\times\Sigma^{1-i}$ such that $\forall k\in D, \;
  \mu(C_k)\neq \mu'_k(C'_k)$. 
\item
  Duplicator responds with states $x_{j+1}^0\in S^0$
  and $x_{j+1}^1\in S^1$ and an $k\in D$
  $x_{j+1}^i\in C_k \iff x_{j+1}^{1-i}\notin C'_k$. The game continues
  from $(x_{j+1}^0,x_{j+1}^1)$.
\end{enumerate}

Duplicator wins if the game is infinite. If a player cannot respond to
a move, they lose the game; hence if for some $j$ one and only one of
the sets $\Tfont{T}_a^0(x_j^0)$ and $\Tfont{T}_a^1(x_j^1)$ is empty,
Spoiler vacuously wins.  On the other hand, Duplicator can win either
by choosing states for which both transition sets are
empty or satisfy $\mu(C)=\mu'(C')$ for all measures $\mu\in
\Tfont{T}_a^0(x^0_j)$, $\mu'\in \Tfont{T}_a^1(x^1_j)$ and measurable
pairs $(C,C')\in\Sigma^0\times\Sigma^1$.

Duplicator's moves correspond to pairs of states $(x_j^0,x_j^1)$
claimed to be ext-state bisimilar, while Spoiler attempts to
show otherwise. Spoiler chooses a label $a$ and a possible behavior
(i.e., a measure) such that no behavior of the other state associated
with the same label can imitate it. The pairs of measurable sets
$(C_k,C'_k)$, supposedly closed under bisimilarity, would testify to
this difference in behavior. However, the reason for this discrepancy
is not disputed
by Duplicator, who claims that some pair $(C_k,C'_k)$ is not closed
under bisimilarity. For this purpose, Duplicator chooses two states
$x_{j+1}^0\in S^0$ and $x_{j+1}^1\in S^1$ such that $x_{j+1}^0\in C_k
\iff x_{j+1}^1\notin C'_k$ and claims that they are bisimilar.

The following lemma is an instance of the Gale-Stewart Theorem
\cite[Thm.~20.1]{Kechris} on infinite games with open or closed payoff sets.
\begin{lemma}
  The game is determined for any initial position $(x_0^0,x_0^1)$.
\end{lemma}
\begin{proof}
  The argument depends solely on the fact that Duplicator wins all
  infinite plays. Suppose Spoiler does not have a winning
  strategy. Then, for each of Spoiler's possible first moves, there must
  exist a move $(x_1^0,x_1^1)$ by Duplicator such that Spoiler does not
  have a winning strategy in the game starting from $(x_1^0,x_1^1)$. We
  repeat this argument to select a second move by Duplicator
  $(x_2^0,x_2^1)$. Continuing in this way, we obtain a winning strategy
  for Duplicator since the play becomes infinite, and thus Duplicator
  wins the game.
\end{proof}

\begin{prop}
  Two states $x_0^0\in S^0$ and $x_0^1\in S^1$ are ext-state bisimilar
  if and only if Duplicator has a winning strategy in the game
  that starts from $(x_0^0,x_0^1)$.
\end{prop}
\begin{proof}
  $(\Rightarrow)$ If Spoiler can carry out the first move, among their
  choice of measurable pairs $(C_k,C'_k)$ there must be one that is
  not closed under external bisimilarity. Duplicator can then choose a
  pair of bisimilar states such that $x_1^0\in C_k \iff x_1^1\notin
  C'_k$ and repeat this strategy in each move.

  $(\Leftarrow)$ We show that $x_0^0\sim^\times_\s x_0^1$. Let $R$ be
  the set of pairs $(x,y)\in S^0\times S^1$ for which Duplicator has a
  winning strategy in the game starting from $(x,y)$. We prove that
  $R$ is an external bisimulation. By contradiction, suppose it is
  not. Without loss of generality, we can assume that there exists
  $(s,t)\in R$ such that $\Tfont{T}_a^1(t) = \{ \mu'_k \mid k\in D \}$
  is nonempty, and there exists $\mu\in \Tfont{T}_a^0(s)$ such that
  for every $k \in D$ there exists an $R$-closed measurable pair
  $(C_k,C'_k)$ such that $\mu(C_k)\neq \mu'_k(C'_k)$.

  We define $i=0, \mu, \{ (C_k,C'_k)\mid k\in D \}$ as Spoiler's first
  move. Regardless of the pair $(s',t')$ played by Duplicator, it must
  be the case that $(s',t')\notin R$ since $(C_k,C'_k)$ is a closed
  pair for every $k$ and, by definition of Duplicator's moves, there
  exists $K\in D$ such that $s'\in C_K \iff t'\notin C'_K$. As
  the game is determined, Spoiler must have a winning strategy from
  $(s',t')$. This defines a winning strategy for Spoiler from $(s
  ,t)\in R$, which is absurd. Therefore, $R$ is an external
  bisimulation.
\end{proof}

For the case of internal bisimilarity, we can use the presented game
with the assumption $\lmp{S}^0=\lmp{S}^1$. In such a situation,
Spoiler is forced to choose coinciding pairs $C=C'$. Otherwise,
Duplicator can select a pair $(x,x)$ for some $x$ such that $x\in
C\iff x\notin C'$ and win the game. Therefore, if we work with a
single NLMP, we can simplify the game to one where Spoiler chooses
measurable sets $C_n$ (instead of pairs of measurable sets), and both
games are equivalent. If the processes do not include nondeterminism,
i.e., they can be thought of as LMP, the resulting game in this case
is identical to the game proposed in \cite{ic19} since Spoiler's move
reduces to a label $a\in L$ and a measurable set $C\in \Sigma$ such
that $\tau_a(x^0,C)\neq \tau_a(x^1,C)$.

\section{Conclusion}
\label{sec:conclusion}

We have seen once more that the intrusion of measure-theoretic
niceties complicate the existing panorama of behavioral
equivalences. Even some of the non-implications have been recent (and
belated) surprises for us: For an example, we had a wrong intuition
that being a (pre)image of a zigzag morphism was one of the strongest
possible connection between processes, but this is now easily refuted
by Example~\ref{exm:zigzag-not-state-domain}.

The role of the direct sum in the study of bisimulations stays a bit
unclear in the same setting, which is more evident by perusing our
Table~\ref{tab:bisim-comparison}. We think that the historically early
Definition~\ref{def:oplus-bisimulation} of $\oplus$-bisimulation is
ill-posed, in that it seems to allow much more room for irregularities
than the rest of the notions studied. It is all the more reasonable
that it has been practically abandoned. On the other hand, we can
state that the operation of “adding a summand” makes a process
smoother, ironing differences out and, in the extreme case, leaving
just the coarser event/cospan/logical equivalence.
Notice that if one is tempted to refine the $\oplus$-bisimulation definition by introducing measurability constraints similar to those in Proposition~\ref{prop:case-R-oplus-is-ext}, it would end up with external bisimulation. 
Also, Example~\ref{exm:oplus-notsbq-cospan} illustrates that in non-measurable scenarios, bisimilarity can become intricate.

We highlight that the most important question left open in this paper
is whether coalgebraic bisimilarity coincides with the one using
plain spans. More probability theory tools seem to be required to
solve this problem.

Concerning NLMP, many questions remain. We were not able to provide
game characterizations for none of the other bisimilarities (neither
for LMP nor NLMP). Also, there is no appropriate coalgebraic presentation
of NLMP. Their category was defined in \cite{2014arXiv1405.7141D}, but
the corresponding functor is not covariant: A new theory of
“contra-coalgebras” might be needed for nondeterministic processes.
We hope that some of these questions will be addressed in the future.

\paragraph*{Acknowledgments}
We warmly thank Giorgio Bacci for his recent comments on this manuscript and for
making us aware of \cite{gburek}. We are very grateful to both anonymous
referees for their detailed reading and many insightful comments. In particular, we are
indebted to Referee 1 for a much simplified Example~\ref{exm:restr-not-bisim},
and for making us realize that an argument involving Bernstein sets (related to
Corollary~\ref{cor:charact-V-final}) was flawed; also Referee 2's observations
helped us to streamline the discussion at the end of Section~\ref{sec:catco-bisimulations}.

\providecommand{\noopsort}[1]{}
\begin{small}\end{small}

\end{document}